\documentclass[letterpaper,journal]{IEEEtran}
\usepackage{amsmath,amsfonts}
\usepackage{amsthm}
\usepackage{algorithmic}
\usepackage{algorithm}

\newtheorem{theorem}{Theorem}
\newtheorem{assumption}{Assumption}

\usepackage{array}
\usepackage{bm}
\usepackage{textcomp}
\usepackage{stfloats}
\usepackage{url}
 \usepackage{multirow}
\usepackage{verbatim}
\usepackage{graphicx}
\graphicspath{{Figs/}}
\usepackage{cite}
\usepackage{hyperref}      
\usepackage{cleveref}
\usepackage[font={small}]{caption}
\usepackage{subcaption}
\hypersetup{hidelinks}

\begin{document}

\title{Safe Planning for Black-Box Systems via MAP Inference with Truncated Diffusion Priors and CBF-Informed Gibbs Likelihoods}

\author{Shanghan Liu, Jitao Sun and Junhui Zhang,~\IEEEmembership{Member,~IEEE,}
		\thanks{This work was supported by the National Natural Science Foundation of China under Grant 62603330. \emph{(Corresponding author:Junhui Zhang)}}
		\thanks{Shanghan Liu and Junhui Zhang are with the School of Robotics and Automation at Nanjing University, Suzhou, Jiangsu 215163, China (e-mail:shanghanliu@smail.nju.edu.cn; junhuiz@nju.edu.cn).}
        \thanks{Jitao Sun is with the School of Mathematics and Statistics, Nantong University, Nantong, 226019, China and also with the School of Mathematical Sciences, Tongji University, Shanghai 200092, China (e-mail: sunjt@tongji.edu.cn). }}



\maketitle

\begin{abstract}
This paper proposes a maximum a posteriori (MAP) inference framework with a truncated diffusion prior for safe planning (TDP-SafeMAP). Under a generalized Bayesian formulation, a criterion based on covariance condition number determines the reverse diffusion truncation timestep, and the clean control distribution predicted at that timestep is approximated by a Gaussian prior encoding the control regularities and uncertainty of expert data. An energy function containing a control barrier function (CBF) constraint violation measure defines a Gibbs task likelihood that incorporates physical constraint information. After fusing the prior and likelihood, the covariance matrix adaptation evolution strategy (CMA-ES), whose core update admits an interpretation as a sampling-based approximation to natural-gradient ascent, approximates the posterior mode using only black-box evaluations. Bézier-curve-based dimensionality reduction and related techniques further improve computational efficiency and task-completion capability. Theoretical analysis establishes an error bound for the approximate MAP solution, an explicit condition for CBF safety preservation, and the equivalence of the MAP problem under prior whitening. Experimental results show that the proposed framework consistently generates safe, smooth, and goal-reaching trajectories, providing an effective and practical statistical inference approach to safe planning for black-box systems.

\end{abstract}


\section{Introduction}

\IEEEPARstart{D}{ata}-driven motion planning has grown rapidly in recent years in areas such as robot control~\cite{Gu2026} and autonomous driving~\cite{Hu2026}. Unlike classical planning methods that rely on accurate dynamic models, data-driven approaches can learn planning policies directly from expert demonstrations or interaction data~\cite{10.1145/3583136}, thereby bypassing the difficulty of modeling complex systems and exhibiting strong generalization in high-dimensional, unstructured environments~\cite{HangYin2021}. Among them, diffusion models, as a powerful class of generative models, have recently been introduced into motion planning and have quickly become a research hotspot~\cite{11244855}. Representative diffusion planners such as Diffuser~\cite{janner2022diffuser} and Diffusion Policy~\cite{chi2025diffusion} can model complex multimodal trajectory distributions, avoiding the over-averaging of expert behavior by unimodal policies; their iterative denoising generation naturally represents long-horizon and smooth control sequences; and they are robust to noise and multimodality in expert data, showing strong generalization in complex environments~\cite{liu2026diffusion}. These advantages make diffusion planners particularly promising for complex, multimodal motion planning tasks, establishing them as an important direction in data-driven planning~\cite{JIN2026115994}.

However, diffusion planners typically focus only on the data distribution and task cost during generation, making it difficult to explicitly incorporate safety constraints and rendering these planners unable to provide reliable safety guarantees at inference time~\cite{11230786}. Moreover, practical systems may exhibit strongly nonlinear, non-convex, and even non-smooth dynamics, whose models are difficult to obtain, making classical trajectory optimization methods that rely on models or gradients difficult to apply directly~\cite{KAZIM2024100931, 10500723}. Consequently, it remains a key challenge to enable diffusion planners to generate safe, smooth, and goal-reaching trajectories under black-box, non-convex, and non-differentiable dynamics~\cite{11247412}.

Explicitly introducing safety constraints into diffusion planning is a popular research direction in generative control, where many methods employ control barrier functions as a means to provide safety information. Control barrier functions (CBFs) offer an effective tool that transforms safety requirements into constraints on control inputs~\cite{10057172, 11267490}. By defining a safe set, a CBF guarantees its forward invariance: in continuous-time systems, it requires the control input to keep the decrease rate of the safety function along the system trajectory below a certain bound~\cite{7782377}; in discrete-time systems, it requires the safety function value at the next instant to be no less than a proportionally decayed value of the current one~\cite{agrawal2017discrete}. This property provides provable safety guarantees for safety-critical systems, and since CBF constraints are convenient to embed into optimization frameworks, CBFs have been widely applied in uncertainty quantification~\cite{11007767, 10872802, 10171460}, adaptive control~\cite{AN2026112782}, model predictive control~\cite{9483029, 9779571, 10644741}, reinforcement learning~\cite{11606602, 10539292, 11578271, 10328440}, etc. In the field of safe planning with diffusion models, SafeDiffuser guarantees the safety of generated trajectories by imposing CBF constraints on the diffusion dynamics and introducing finite-time diffusion invariance~\cite{xiao2025safe}, CoBL-Diffusion guides reverse diffusion via CBF and control Lyapunov function (CLF) reward gradients~\cite{MizutaLeung2024}, Constrained Diffuser formulates constrained diffusion generation as a constrained sampling problem and adopts discrete CBFs as constraints~\cite{NEURIPS2025_31fb284a}, and PC-Diffuser enforces safety, dynamic feasibility, and path consistency via a path-consistent CBF safety filter within the denoising loop~\cite{ku2026pc}. However, under black-box and non-differentiable dynamics, it is difficult to compute the gradient of the CBF condition, and classical CBF-QP or gradient-based policy optimization becomes inapplicable~\cite{GARG2024100945, 9681233}. How to combine the safety guarantees of CBFs with the data-driven capability of diffusion priors and solve the resulting problem in a gradient-free manner under black-box dynamics remains an open problem.

Bayesian inference is a probabilistic method in which the prior distribution carries knowledge or assumptions about the parameters, observed data enters the inference through the likelihood function, and the two are unified by the posterior distribution~\cite{Ghahramani2015}. From a probabilistic perspective, the problem of integrating safety constraints into diffusion planning naturally corresponds to Bayesian inference: the multimodal distribution learned from expert data by the diffusion model is essentially a prior belief over control sequences~\cite{11097366}, while the CBF safety constraints and task requirements are encoded by a Gibbs generalized task likelihood, which evaluates how well each candidate solution satisfies the task specification. Through this generalized Bayesian update, the data-driven information provided by the diffusion model and the physical information provided by the task energy can be naturally fused.

In existing research, the reverse denoising process of diffusion models has been shown to exhibit a speciation time~\cite{biroli2024dynamical}, after which the modes have largely separated and subsequent steps mainly perform fine adjustments within the selected mode. Since fusing CBF information requires gradient-free optimization, the fine-adjustment part can be completed simultaneously in this process, and the primary role of the diffusion prior is to leverage its generative diversity to select the most appropriate behavior mode~\cite{diffusiondrive}. Directly performing posterior computation and optimization after identifying that the diffusion inference has completed mode selection can avoid the mode-averaging problem caused by running the diffusion model multiple times and forcibly averaging the results, while significantly reducing the time burden of diffusion inference.

Since the model is black-box, the gradient of the CBF condition cannot be computed directly, necessitating gradient-free optimization methods. CMA-ES~\cite{hansen2001completely}, as a gradient-free evolution strategy, can approach the optimal solution without gradients. Compared with random search, finite differences, Bayesian optimization and other gradient-free methods, CMA-ES offers advantages such as adaptive covariance, rotational invariance~\cite{HANSEN20115755}, and robust dimension-dependent default parameter settings~\cite{aguirre2026covariance}. From a theoretical perspective, its update mechanism can be interpreted under the information-geometric optimization (IGO) and MAP-IGO frameworks as natural gradient ascent on the Gaussian statistical manifold~\cite{akimoto2012theoretical, ollivier2017information, hamano2024natural}. Natural gradient ascent itself can be employed to solve MAP estimation problems, which is better than ordinary gradient ascent because of its parameterization invariance and its use of the Fisher information matrix as the metric, rescaling the update by the actual change in the distribution rather than in the parameters~\cite{Wu_2024}. This establishes a theoretical bridge between stochastic search optimization and Bayesian point estimation, revealing an intrinsic connection between CMA-ES and MAP estimation in Bayesian inference. Based on this property, we choose CMA-ES as the posterior solver to approximate the Bayesian posterior mode in a low-dimensional control-point space obtained via Bézier curves, thereby explicitly introducing CBF safety constraints while preserving the multimodal expressiveness of the diffusion prior.

The above observations collectively point to a unified planning framework: using the mode selection learned from expert demonstrations by a data-driven method as a distribution prior, constructing a likelihood function from path requirements such as safety, merging the two into a posterior, and then performing maximum a posteriori estimation using a gradient-free optimization method. To this end, we propose TDP-SafeMAP, a safe planning method for black-box nonlinear systems. The main contributions of this paper are as follows:

\begin{itemize}
    \item 
    \textbf{A generalized Bayesian inference framework for safe planning of black-box systems.} 
    We project the predicted control distribution at a truncated diffusion timestep onto a Gaussian distribution to construct the prior, and define a Gibbs generalized likelihood via a task energy that incorporates a CBF constraint-violation measure. Together, they form a Gibbs posterior, thereby formulating safe planning as a MAP estimation problem. The framework decouples the data prior from task constraints, so the diffusion model need not be retrained when requirements such as safety and smoothness change.
    
    \item 
    \textbf{An adaptive diffusion truncation strategy.} 
    Based on the correspondence between the covariance condition number and the principal-direction Bayes risk, the method identifies the transition from unresolved mode uncertainty to within-mode refinement from the peak of the second-order difference of the condition number, thereby automatically determining the truncation timestep near mode separation. This process can be performed offline using the dataset information and the relevant parameters of the diffusion model, thereby improving the computational efficiency of the algorithm.
    
    \item 
    \textbf{A theoretical analysis of algorithmic robustness and related properties.}
    We analyze how prior approximation and finite CMA-ES accuracy affect the MAP solution and derive a sufficient condition for preserving CBF safety, characterize the information-geometric compatibility between CMA-ES and posterior search, and establish the exact equivalence of the original and prior-whitened MAP problems.
\end{itemize}

The remainder of this paper is organized as follows. Section II introduces the background knowledge required for the relevant methods. Section III provides the problem formulation and the Bayesian inference framework. Section IV describes the core methodology of TDP-SafeMAP in detail, covering the truncated acquisition of the diffusion prior, the construction of the likelihood function, and the inference of the posterior distribution. Section V provides a theoretical analysis of the proposed method. Section VI presents experimental results and analysis. Finally, Section VII concludes the paper and discusses future work.

\section{Preliminaries}
\label{sec:prelim}

\subsection{Bayesian Inference}
\label{sec:bayes_intro}

Bayesian inference is a probabilistic framework for reasoning under uncertainty. It regards all unknown quantities as random variables and describes their uncertainty through probability distributions~\cite{van2021bayesian}. This methodology naturally integrates prior knowledge---the beliefs about the parameters before observing any evidence---with the information conveyed by the observed outcomes, enabling a coherent update of belief within a unified probabilistic calculus.

Let $\boldsymbol{\theta}$ denote the vector of unknown parameters of interest, and let $\mathcal{D}$ denote the observed outcomes. The cornerstone of Bayesian inference is Bayes' theorem, which defines the posterior distribution as
\begin{equation*}
p(\boldsymbol{\theta}|\mathcal{D}) = \frac{p(\mathcal{D}|\boldsymbol{\theta}) \, p(\boldsymbol{\theta})}{p(\mathcal{D})}
\propto p(\mathcal{D}|\boldsymbol{\theta}) \, p(\boldsymbol{\theta}),
\end{equation*}
where:
\begin{itemize}
\item $p(\boldsymbol{\theta})$ is the \emph{prior distribution}, which quantifies the initial belief about $\boldsymbol{\theta}$ before observing any outcomes. The prior may be informative, derived from expert knowledge or physical constraints, or non-informative to minimize its influence on the posterior.
\item $p(\mathcal{D}|\boldsymbol{\theta})$ is the \emph{likelihood function}. For a fixed value of $\boldsymbol{\theta}$, it gives the probability of observing the outcomes $\mathcal{D}$; once $\mathcal{D}$ is observed and held fixed, it is viewed as a function of $\boldsymbol{\theta}$ that measures how well each candidate parameter value explains the data.
\item $p(\mathcal{D}) = \int p(\mathcal{D}|\boldsymbol{\theta}) p(\boldsymbol{\theta}) \, \mathrm{d}\boldsymbol{\theta}$ is the \emph{marginal likelihood} (or evidence), which serves as a normalizing constant ensuring that the posterior integrates to unity.
\item $p(\boldsymbol{\theta}|\mathcal{D})$ is the \emph{posterior distribution}, which synthesizes the prior and the likelihood, reflecting the updated understanding of the parameters after observing the outcomes. 
\end{itemize}

In many practical engineering problems, the posterior distribution is analytically intractable, primarily because the marginal likelihood cannot be evaluated in closed form in the vast majority of cases~\cite{8588399}. A common remedy is not to compute the full posterior, but to seek its mode, i.e., the MAP estimate:
\begin{equation*}
\boldsymbol{\theta}_{\text{MAP}} = \arg\max_{\boldsymbol{\theta}} p(\boldsymbol{\theta}|\mathcal{D})
= \arg\max_{\boldsymbol{\theta}} p(\mathcal{D}|\boldsymbol{\theta}) \, p(\boldsymbol{\theta}),
\end{equation*}
where the normalizing constant $p(\mathcal{D})$ has been dropped because it does not depend on $\boldsymbol{\theta}$. In this way, the intractable marginal likelihood is bypassed entirely, and the inference problem reduces to maximizing the unnormalized posterior.

Beyond ordinary data likelihoods, energy-based engineering requirements can be incorporated through a Gibbs generalized likelihood. For notational continuity, we retain the conventional likelihood notation and write
\begin{equation*}
p(\mathcal{D}|\boldsymbol{\theta})
\propto
\exp\bigl(-E(\boldsymbol{\theta})\bigr).
\end{equation*}
Here, $p(\mathcal{D}|\boldsymbol{\theta})$ is understood as a nonnegative compatibility factor rather than an explicitly specified data-generating distribution over $\mathcal{D}$. Combining this factor with the prior defines the normalized Gibbs posterior:
\begin{equation*}
p(\boldsymbol{\theta}|\mathcal{D})
=
\frac{
    \exp\bigl(-E(\boldsymbol{\theta})\bigr)
    p(\boldsymbol{\theta})
}{
    \int
    \exp\bigl(-E(\boldsymbol{\theta}')\bigr)
    p(\boldsymbol{\theta}')\,
    \mathrm{d}\boldsymbol{\theta}'
}.
\end{equation*}

In many practical engineering problems, this posterior distribution is analytically intractable because the normalizing integral cannot be evaluated in closed form~\cite{8588399}. A common remedy is not to compute the full posterior, but to seek its mode, i.e., the maximum a posteriori (MAP) estimate. Taking the negative logarithm, the corresponding MAP estimate is
\begin{equation*}
\boldsymbol{\theta}_{\mathrm{MAP}}
=
\arg\min_{\boldsymbol{\theta}}
\left[
    E(\boldsymbol{\theta})
    -
    \log p(\boldsymbol{\theta})
\right].
\end{equation*}
When the prior is Gaussian, the negative log-prior becomes a quadratic regularizer. Consequently, generalized Bayesian inference directly maps engineering objectives into a deterministic energy-minimization problem, enabling the principled integration of physical constraints that lack an explicit probabilistic data model.

\subsection{Diffusion Models}
\label{sec:diffusion_models}

Diffusion models are a class of generative models that learn to approximate a data distribution by reversing a gradual noising process~\cite{ho2020denoising}. The framework is built upon two Markov chains: a fixed forward process that gradually adds noise to a data sample, transforming it into a simple noise distribution (typically a standard Gaussian), and a learned reverse process that reconstructs data by iteratively removing noise. Because the true reverse transition is not available in closed form, it is approximated by a neural network trained to predict the noise introduced at each step. This section reviews the mathematical framework of diffusion models used in this work, including the forward noising process, the training of the network, and the reverse denoising process.
\paragraph{Forward Diffusion Process}
Given a data sample $\mathbf{x}_0 \sim q(\mathbf{x}_0)$, the forward process gradually adds Gaussian noise through a Markov chain of $T$ steps, producing $\mathbf{x}_1, \dots, \mathbf{x}_T$. The transition kernel is
\begin{equation*}
q(\mathbf{x}_t | \mathbf{x}_{t-1}) = \mathcal{N}\left(\mathbf{x}_t; \sqrt{1-\beta_t}\, \mathbf{x}_{t-1}, \beta_t I\right),
\end{equation*}
where $\beta_t \in (0,1)$ is a predefined noise schedule parameter. Let $\alpha_t = 1 - \beta_t$ and $\bar{\alpha}_t = \prod_{s=1}^t \alpha_s$. We adopt the cosine schedule~\cite{nichol2021improved}, which defines $\bar{\alpha}_t$ directly as
\begin{equation*}
\bar{\alpha}_t = \frac{f(t)}{f(0)}, \quad f(t) = \cos\left(\frac{t/T + s}{1 + s} \cdot \frac{\pi}{2}\right)^2,
\end{equation*}
where $s$ is a small offset that keeps $\beta_t$ from being too small near $t = 0$, 
thereby avoiding an overly steep change of $\bar{\alpha}_t$ at the beginning of the forward process.

Using these parameters, the noisy sample at any timestep can be computed directly from $\mathbf{x}_0$:
\begin{equation*}
\mathbf{x}_t = \sqrt{\bar{\alpha}_t}\, \mathbf{x}_0 + \sqrt{1-\bar{\alpha}_t}\, \boldsymbol{\varepsilon}, \quad \boldsymbol{\varepsilon} \sim \mathcal{N}(0, I),
\end{equation*}
which gives $q(\mathbf{x}_t | \mathbf{x}_0) = \mathcal{N}(\mathbf{x}_t; \sqrt{\bar{\alpha}_t}\, \mathbf{x}_0, (1-\bar{\alpha}_t)I)$. When $\bar{\alpha}_T \approx 0$, $\mathbf{x}_T$ approximates a standard Gaussian.

The forward process gradually provides learnable noising trajectories by adding noise for training the noise prediction network in the reverse process, enabling the model to learn how to reconstruct data from noise.

\paragraph{Training the Noise Prediction Network}
The reverse kernels $q(\mathbf{x}_{t-1} | \mathbf{x}_t)$ are intractable and thus need to be approximated. A noise prediction network $\boldsymbol{\varepsilon}_\theta(\mathbf{x}_t, t)$ is trained to minimize
\begin{equation*}
\mathcal{L} = \mathbb{E}_{t, \mathbf{x}_0, \boldsymbol{\varepsilon}} \left[ \|\boldsymbol{\varepsilon} - \boldsymbol{\varepsilon}_\theta(\mathbf{x}_t, t)\|^2 \right],
\end{equation*}
where $t$ is uniformly sampled from $\{1, \dots, T\}$, $\mathbf{x}_0 \sim q(\mathbf{x}_0)$, and $\boldsymbol{\varepsilon} \sim \mathcal{N}(0, I)$. This objective is derived from the variational lower bound and teaches the network to denoise corrupted data. We choose to predict the noise at each step rather than the posterior mean, as this parameterization is empirically more stable to train and yields better sample quality.

\paragraph{Reverse Denoising Process}
Once the noise prediction network is trained, generation can be performed by starting from pure noise $\mathbf{x}_T \sim \mathcal{N}(0, I)$ and iteratively denoising through the learned reverse transitions. Each reverse step is parameterized as
\begin{equation*}
p_\theta(\mathbf{x}_{t-1} | \mathbf{x}_t) = \mathcal{N}\left(\mathbf{x}_{t-1}; \mu_\theta(\mathbf{x}_t, t), \Sigma_\theta(\mathbf{x}_t, t)\right),
\end{equation*}
where the mean is computed from the noise prediction network as
\begin{equation*}
\mu_\theta(\mathbf{x}_t, t) = \frac{1}{\sqrt{\alpha_t}} \left( \mathbf{x}_t - \frac{1 - \alpha_t}{\sqrt{1 - \bar{\alpha}_t}} \boldsymbol{\varepsilon}_\theta(\mathbf{x}_t, t) \right),
\end{equation*}
and the variance $\Sigma_\theta(\mathbf{x}_t, t)$ is typically fixed to $\tilde{\beta}_t I$ with $\tilde{\beta}_t = \frac{1 - \bar{\alpha}_{t-1}}{1 - \bar{\alpha}_t} \beta_t$, which corresponds to the true posterior variance of the forward process. Iteratively applying these reverse steps for $t = T, \dots, 1$ gradually removes noise and yields a generated sample $\mathbf{x}_0$.

\subsection{CMA-ES Algorithm}
\label{sec:cmaes_intro}

The Covariance Matrix Adaptation Evolution Strategy is a stochastic optimization algorithm for continuous black-box problems. Its core idea is to maintain a multivariate normal distribution $\mathcal{N}(\mathbf{m}^{(g)}, (\sigma^{(g)})^2 \mathbf{C}^{(g)})$ at each generation, from which $\lambda$ candidate solutions are sampled. Using the weighted information of the top $\mu_e$ elite samples, the distribution parameters $\mathbf{m}$, $\mathbf{C}$, and the step size $\sigma$ are iteratively updated to steer the sampling distribution toward the optimum of the objective function.

The core parameter updates are as follows. The mean is updated via weighted recombination of elite samples:
\begin{equation*}
\mathbf{m}^{(g+1)} = \sum_{i=1}^{\mu_e} w_i \mathbf{x}_{i:\lambda}^{(g)},\quad \sum_{i=1}^{\mu_e} w_i = 1,\; w_1 \geq \cdots \geq w_{\mu_e} > 0,
\end{equation*}
where $\mathbf{x}_{i:\lambda}$ denotes the $i$-th best individual in generation $g$ according to fitness ranking, and the weights $w_i$ are typically set according to a logarithmic ranking rule.

The covariance matrix is updated via a combination of rank-one and rank-$\mu$ update strategies:
\begin{equation}
\begin{split}
\mathbf{C}^{(g+1)} = &(1-c_1-c_\mu)\mathbf{C}^{(g)} + c_1 \mathbf{p}_c^{(g+1)}(\mathbf{p}_c^{(g+1)})^\top \\
&+ c_\mu \sum_{i=1}^{\mu_e} w_i \mathbf{y}_{i:\lambda}^{(g)}(\mathbf{y}_{i:\lambda}^{(g)})^\top,
\end{split}
\label{eq:cmaes_cov}
\end{equation}
where $\mathbf{y}_{i:\lambda} = (\mathbf{x}_{i:\lambda}-\mathbf{m}^{(g)})/\sigma^{(g)}$ denotes the normalized deviation of the $i$-th best sample from the mean. The coefficients $c_1 \approx 2/n^2$ and $c_\mu \approx \mu_e/n^2$ are the covariance learning rates. The evolution path $\mathbf{p}_c$ accumulates historical mean shifts from successive generations and guides the covariance matrix adaptation toward promising search directions.

The step size $\sigma$ is adapted by cumulative step-size adaptation (CSA), which maintains a second evolution path $\mathbf{p}_\sigma$ and increases or decreases $\sigma$ according to whether consecutive search steps are positively correlated.

All internal parameters of CMA-ES are determined by predefined formulas based solely on the problem dimension $n$, requiring no manual tuning; together with its rotationally invariant covariance update, this makes CMA-ES robust and easily applicable in practical engineering scenarios.

\subsection{Control Barrier Function}
\label{sec:cbf_intro}
A control barrier function is an effective tool in safety-critical control. It ensures the forward invariance of a given set by transforming safety constraints into constraints on control inputs, thereby keeping the system state always within the prescribed safe region when the corresponding condition is satisfied~\cite{ames2019control}.
\paragraph{Continuous-Time Formulation}
In the original continuous-time setting, consider the control-affine system $\dot{\mathbf{x}} = f(\mathbf{x}) + g(\mathbf{x})\mathbf{u}$, where $\mathbf{x} \in \mathbb{R}^n$ is the state and $\mathbf{u} \in \mathcal{U} \subset \mathbb{R}^m$ is the control input. Let the safe set $\mathcal{C}$ be defined as the superlevel set of a continuously differentiable function $h: \mathbb{R}^n \to \mathbb{R}$:
\begin{equation*}
\mathcal{C} = \{\mathbf{x} \in \mathbb{R}^n : h(\mathbf{x}) \geq 0\},
\end{equation*}
with boundary $\partial\mathcal{C} = \{\mathbf{x} : h(\mathbf{x}) = 0\}$. The value of $h(\mathbf{x})$ can be interpreted as the current safety margin of the system, indicating proximity to the unsafe region.

A function $h$ is said to be a zeroing control barrier function if there exists an extended class $\mathcal{K}_\infty$ function $\alpha(\cdot)$ (strictly increasing with $\alpha(0)=0$ and $\lim_{r\to\pm\infty}\alpha(r)=\pm\infty$) such that for all $\mathbf{x} \in \mathcal{C}$,
\begin{equation*}
\sup_{\mathbf{u} \in \mathcal{U}} \left[ L_f h(\mathbf{x}) + L_g h(\mathbf{x}) \mathbf{u} \right] \geq -\alpha(h(\mathbf{x})),
\end{equation*}
where $L_f h(\mathbf{x}) = \frac{\partial h}{\partial \mathbf{x}} f(\mathbf{x})$ and $L_g h(\mathbf{x}) = \frac{\partial h}{\partial \mathbf{x}} g(\mathbf{x})$ are the Lie derivatives of $h$ along $f$ and $g$, respectively. This condition ensures that when the system approaches the boundary, there exists a control input that prevents the safety margin from decreasing too fast. By Nagumo's theorem, this guarantees forward invariance of $\mathcal{C}$: if $\mathbf{x}(0) \in \mathcal{C}$, then $\mathbf{x}(t) \in \mathcal{C}$ for all $t \geq 0$.

\paragraph{Discrete-Time Formulation}
In practice, robotic controllers are typically implemented in discrete time with sampling period $\Delta t$. Accordingly, the discrete-time counterpart of the CBF condition is defined as
\begin{equation}
h(\mathbf{x}_{k+1}) \geq (1 - \gamma) \, h(\mathbf{x}_k), \quad 0 < \gamma \leq 1,
\label{eq:cbf_condition_disc}
\end{equation}
where $\gamma \in (0, 1]$ plays an analogous role to the continuous-time class $\mathcal{K}_\infty$ function, controlling the decay rate of the safety margin. When $h(\mathbf{x}_k) = 0$, condition (\ref{eq:cbf_condition_disc}) reduces to $h(\mathbf{x}_{k+1}) \geq 0$, requiring that the next state remains inside the safe set. This discrete condition guarantees that if $\mathbf{x}_k \in \mathcal{C}$, then $\mathbf{x}_{k+1} \in \mathcal{C}$, thereby enforcing forward invariance along the sampled trajectory.

\subsection{Bézier Curves}
\label{sec:bezier_intro}

Bézier curves are parametric curves widely used in computer graphics, path planning, and trajectory generation due to their intuitive geometric control and smoothness properties. A Bézier curve of degree $n$ is defined as a linear combination of $n+1$ control points $\mathbf{B}_i \in \mathbb{R}^d$, where $d$ is the spatial dimension, using Bernstein basis polynomials~\cite{farouki2012bernstein}:
\begin{equation*}
\mathbf{C}(t) = \sum_{i=0}^{n} \mathbf{B}_i B_{i,n}(t), \quad t \in [0,1],
\end{equation*}
where $B_{i,n}(t) = \binom{n}{i} t^i (1-t)^{n-i}$ is the $i$-th Bernstein polynomial of degree $n$. The control points $\mathbf{B}_0, \mathbf{B}_1, \dots, \mathbf{B}_n$ form the control polygon, which geometrically approximates the curve shape.

Bézier curves possess several important properties that make them suitable for trajectory planning. First, they exhibit endpoint interpolation: $\mathbf{C}(0) = \mathbf{B}_0$ and $\mathbf{C}(1) = \mathbf{B}_n$, meaning the curve passes through the first and last control points. Second, the curve lies entirely within the convex hull of its control points, providing a natural guarantee of spatial bounds. Third, the derivatives at the endpoints are determined solely by the adjacent control points:
\begin{equation*}
\dot{\mathbf{C}}(0) = n(\mathbf{B}_1 - \mathbf{B}_0), \quad \dot{\mathbf{C}}(1) = n(\mathbf{B}_n - \mathbf{B}_{n-1}),
\end{equation*}
which allows direct enforcement of initial and terminal velocity constraints in trajectory generation.

The derivative of a Bézier curve is itself a lower-degree Bézier curve, expressed as
\begin{equation*}
\dot{\mathbf{C}}(t) = n \sum_{i=0}^{n-1} (\mathbf{B}_{i+1} - \mathbf{B}_i) B_{i,n-1}(t).
\end{equation*}
This property facilitates the derivation of higher-order derivatives and enables the incorporation of continuity constraints (e.g., $C^1$ and $C^2$ smoothness) when concatenating multiple Bézier segments. Moreover, the boundedness of the Bernstein basis on $t \in [0,1]$ ensures numerical stability, making Bézier curves a computationally efficient and mathematically tractable choice for constrained trajectory optimization.

\section{Problem Formulation}

Consider the discrete-time nonlinear control system:
\begin{equation}
\mathbf{x}_{k+1} = f(\mathbf{x}_k, \mathbf{u}_k), \quad k = 0, 1, \ldots, H - 1,
\label{eq:system}
\end{equation}
where $\mathbf{x}_k \in \mathcal{X} \subseteq \mathbb{R}^n$ is the state, $\mathbf{u}_k \in \mathcal{U} \subseteq \mathbb{R}^m$ is the control input, and $H$ is the planning horizon. The initial state $\mathbf{x}_0 = \mathbf{x}_{\mathrm{init}}$ is given. Let $\mathbf{U} = [\mathbf{u}_0^\top, \mathbf{u}_1^\top, \ldots, \mathbf{u}_{H-1}^\top]^\top \in \mathbb{R}^{mH}$ denote the complete control sequence, which uniquely determines the state trajectory $\mathbf{X}(\mathbf{U}) = [\mathbf{x}_1^\top, \ldots, \mathbf{x}_H^\top]^\top$ through the recursion in \eqref{eq:system}. The dynamics $f(\cdot,\cdot)$ are assumed to be black-box with no differentiability requirement, and the state and control constraints $\mathbf{x}_k \in \mathcal{X}$, $\mathbf{u}_k \in \mathcal{U}$ are required to hold for all $k$.

The planning task is to determine a control sequence $\mathbf{U}$ that drives the system from $\mathbf{x}_{\mathrm{init}}$ to the target $\mathbf{x}_{\mathrm{goal}}$, satisfying the following requirements:
\begin{enumerate}
    \item[(i)] \emph{Plausible path:} the resulting trajectory should be consistent with the distribution of expert demonstrations, so that the planned motion follows a task-appropriate mode rather than an arbitrary feasible one. This requirement is encoded by a truncated diffusion prior over control sequences.
    \item[(ii)] \emph{Safety:} the resulting trajectory $\mathbf{X}(\mathbf{U})$ remains within the safe set $\mathcal{C}$ throughout the horizon, i.e., the discrete-time CBF condition \eqref{eq:cbf_condition_disc} holds at every step.
    \item[(iii)] \emph{Control sequence requirements:} the control sequence should satisfy several execution-level requirements, including smoothness (no abrupt changes that would cause saturation or chattering), low energy consumption (bounded control effort), etc.
\end{enumerate}

We choose the control sequence $\mathbf{U}$ rather than the state trajectory $\mathbf{X}$ as the decision variable, since a physical system executes control inputs and a plan expressed in the control space is thus directly executable, whereas a generated state sequence may not be dynamically consistent, i.e., it may not be reachable by any admissible control. Moreover, under black-box dynamics the control sequence is the only quantity that can be directly applied to the forward simulator, so that the safety, task, and smoothness terms can all be evaluated in a single forward rollout without an additional inverse-dynamics step to recover the controls.

The problem is challenging because the data-driven requirement and the physical/task requirements are heterogeneous in nature and must be fused in a principled manner rather than treated as independent objectives. Besides, the dynamics $f$ are black-box and non-differentiable while the per-step CBF condition renders the feasible set non-convex, making the planning problem a non-convex, gradient-inaccessible optimization over a high-dimensional decision variable $\mathbf{U} \in \mathbb{R}^{mH}$.

\section{Method}
\label{sec:method}

\begin{figure*}[!t]
\centering
\includegraphics[width=\linewidth]{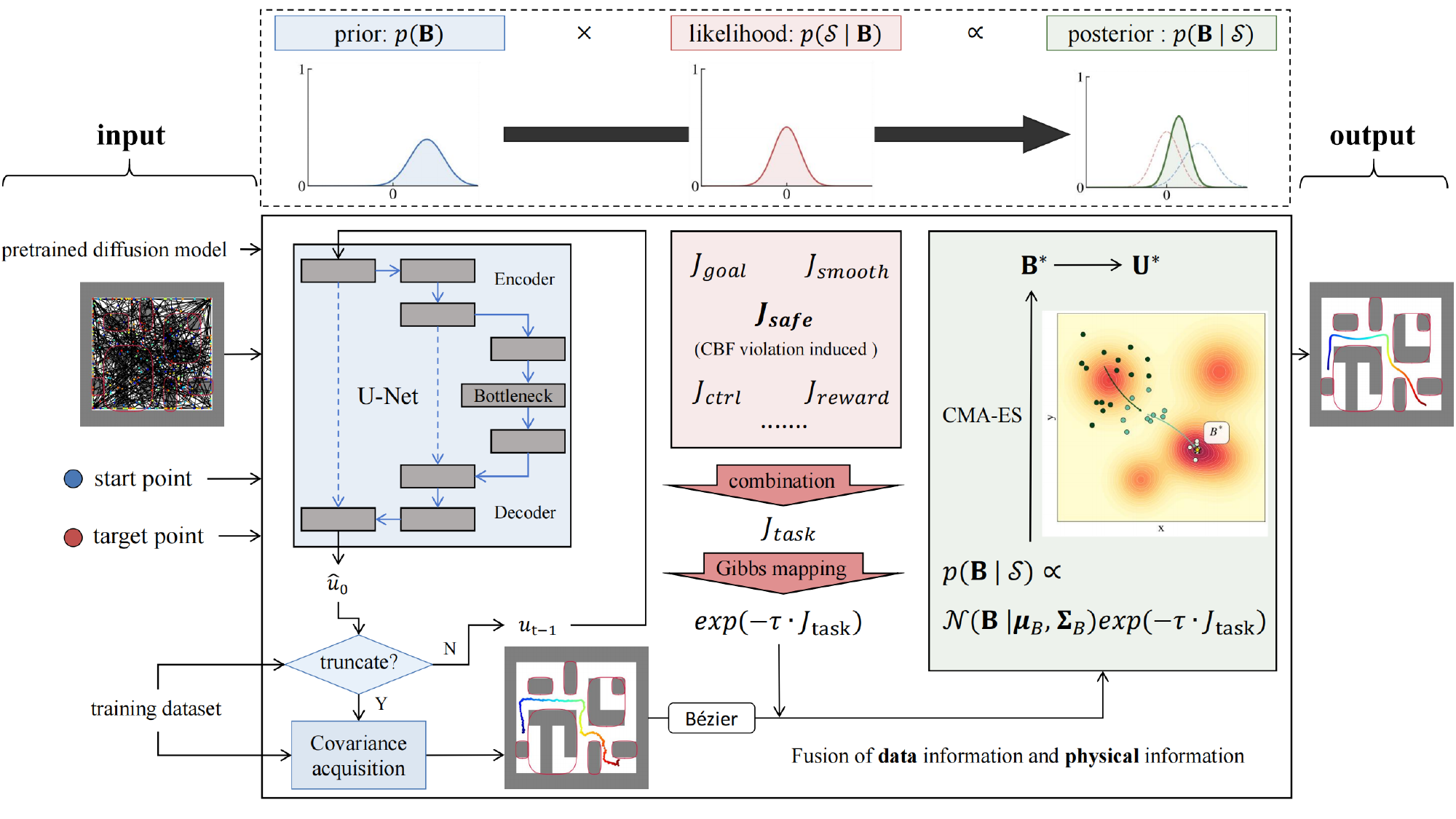}
\caption{Overall framework of TDP-SafeMAP. The truncated diffusion model
constructs a Gaussian prior in the control sequence space, which is then
mapped into the B\'ezier control-point space, while the task energy induces
a Gibbs generalized likelihood. Their product defines a Gibbs posterior,
whose mode is approximately computed by CMA-ES. The one-dimensional
distributions are shown schematically.}
\label{fig1}
\end{figure*}

\subsection{Bayesian Reformulation of the Planning Problem}

We formulate this planning problem under a generalized Bayesian inference framework as a Maximum a Posteriori (MAP) estimation problem. Specifically:

\textbf{Prior distribution}: $p(\mathbf{B}) = \mathcal{N}(\boldsymbol{\mu}_B, \boldsymbol{\Sigma}_B)$, constructed from offline expert data by a truncated diffusion model, encoding the data-driven belief about control sequences conditioned on the task environment.

\textbf{Generalized Likelihood Function}: The term \(p(\mathcal{S} \mid \mathbf{B}) \propto \exp(-\tau \cdot J_{\mathrm{task}}(\mathbf{B}))\) is defined as a Gibbs generalized task likelihood, where \(\tau > 0\) is the inverse temperature and \(\mathcal{S}\) reflects the degree to which the candidate solution satisfies the task specification.

\textbf{Posterior distribution}: $p(\mathbf{B} \mid \mathcal{S}) \propto p(\mathcal{S} \mid \mathbf{B}) \, p(\mathbf{B})$ defines the Gibbs posterior that fuses the diffusion prior with the task likelihood. Its negative logarithm yields the MAP estimation objective:
\begin{equation}
\begin{split}
\mathbf{B}_{\mathrm{MAP}} = \arg\min_{\mathbf{B}} \Big[ 
& \tau \cdot J_{\mathrm{task}}(\mathbf{B}) \\
& + \frac{1}{2} (\mathbf{B} - \boldsymbol{\mu}_B)^\top \boldsymbol{\Sigma}_B^{-1} (\mathbf{B} - \boldsymbol{\mu}_B) \Big].
\end{split}
\label{eq:map_objective_posterior}
\end{equation}

\textbf{Posterior solver}: since $J_{\mathrm{task}}$ relies on black-box forward simulation and is non-differentiable, we adopt a gradient-free evolution strategy to solve \eqref{eq:map_objective_posterior}.

\subsection{Construction of Truncated Diffusion Prior}

To embed the complex trajectory distribution captured by the pretrained diffusion model into the Bayesian inference framework and render it compatible with the subsequent CMA-ES-based MAP estimation, we propose an offline construction method for a truncated diffusion prior. The core idea is to terminate the reverse denoising process at an intermediate timestep \( t_{\mathrm{stop}} \) and approximate the predicted distribution of the clean data \( \mathbf{x}_0 \) conditioned on \( \mathbf{x}_{t_{\mathrm{stop}}} \) as a Gaussian \( \mathcal{N}(\boldsymbol{\mu}_U, \boldsymbol{\Sigma}_U) \). The primary challenge lies in preserving the multimodal mode-selection capability of the diffusion model while obtaining a compact and accurate Gaussian parameterization.

In the reverse denoising dynamics of diffusion models, there exists a critical phase termed \textit{speciation}. Prior to this phase, the reverse process primarily recovers the macroscopic global structure from pure noise, with individual latent modes remaining undifferentiated. After crossing this phase, the conditional modes have largely separated, and subsequent denoising steps mainly perform local refinements within each mode. Motivated by this observation, our truncation strategy aims to precisely locate this speciation time and extract the Gaussian approximation at that point. At this stage, the output distribution of the diffusion model exhibits clear mean selection, while the distribution within a single mode can typically be approximated by a Gaussian distribution, thereby providing theoretical support for the Gaussian truncation approximation.

\subsubsection{Acquisition of Mean Prediction}

Given the noisy sample \(\mathbf{x}_t\) at timestep \(t\), Tweedie's formula~\cite{efron2011tweedie} 
provides an estimate of the clean data \(\mathbf{x}_0\) as
\begin{equation}
\boldsymbol{\mu}_{0|t} = \frac{\mathbf{x}_t - \sqrt{1-\bar{\alpha}_t} \, \boldsymbol{\varepsilon}_\theta(\mathbf{x}_t, t)}{\sqrt{\bar{\alpha}_t}},
\label{eq:tweedie}
\end{equation}
where \(\boldsymbol{\varepsilon}_\theta(\cdot, \cdot)\) is the pretrained noise prediction network. 
This estimate corresponds to the mean of the Gaussian approximation \(p(\mathbf{x}_0 \mid \mathbf{x}_t) \approx \mathcal{N}(\boldsymbol{\mu}_{0|t}, \boldsymbol{\Sigma}_{0|t})\). 
At the truncation timestep \(t_{\mathrm{stop}}\), we set 
\(\boldsymbol{\mu}_U = \boldsymbol{\mu}_{0|t_{\mathrm{stop}}}\). 

\subsubsection{Acquisition of Covariance Prediction}

Let the random variable in the original control sequence space be denoted as \( \mathbf{U} \in \mathbb{R}^{D} \), where \( D = m \times H \). Given the training dataset \( \{\mathbf{u}_{\mathrm{data}}^{(i)}\}_{i=1}^{M} \), we first compute the empirical mean and covariance matrix:
\begin{equation*}
\bar{\mathbf{u}}_{\mathrm{data}} = \frac{1}{M}\sum_{i=1}^{M}\mathbf{u}_{\mathrm{data}}^{(i)},
\end{equation*}
\begin{equation*}
\boldsymbol{\Sigma}_{\mathrm{data}} = \frac{1}{M - 1}\sum_{i=1}^{M}(\mathbf{u}_{\mathrm{data}}^{(i)} - \bar{\mathbf{u}}_{\mathrm{data}})(\mathbf{u}_{\mathrm{data}}^{(i)} - \bar{\mathbf{u}}_{\mathrm{data}})^{\top}.
\end{equation*}
To ensure numerical stability, we apply an isotropic regularization to the covariance matrix: \( \boldsymbol{\Sigma}_{\mathrm{data}} \leftarrow \boldsymbol{\Sigma}_{\mathrm{data}} + \epsilon \mathbf{I} \), where \( \epsilon = 10^{-6} \cdot \mathrm{tr}(\boldsymbol{\Sigma}_{\mathrm{data}}) / D \).

Following the Free Hunch method \cite{rissanen2025free}, at any timestep \( t \) in the reverse process, the predicted distribution of the clean data \( \mathbf{x}_0 \) conditioned on the current noisy sample \( \mathbf{x}_t \) can be approximated as a Gaussian: \( p(\mathbf{x}_0 \mid \mathbf{x}_t) \approx \mathcal{N}(\boldsymbol{\mu}_{0|t}, \boldsymbol{\Sigma}_{0|t}) \). The mean \( \boldsymbol{\mu}_{0|t} \) follows directly from Tweedie's formula, whereas the covariance is obtained through a separate recursion, as described below.

Let \( \bar{\alpha}_t \) denote the noise accumulation coefficient in the forward process, and define \( \rho_t = \bar{\alpha}_t / (1 - \bar{\alpha}_t) \). Setting the covariance at the initial timestep \( t_{\max} \) as the data covariance \( \boldsymbol{\Sigma}_{0|t_{\max}} = \boldsymbol{\Sigma}_{\mathrm{data}} \), when stepping back from \( t \) to \( t' \) (where \( t' < t \)), the inverse covariance matrix satisfies the recursion:
\begin{equation}
\boldsymbol{\Sigma}_{0|t'}^{-1} = \boldsymbol{\Sigma}_{0|t}^{-1} + \left( \rho_{t'} - \rho_t \right) \mathbf{I}.
\label{eq:trunc_cov_recursion}
\end{equation}
This recursion depends only on the noise schedule \( \{\bar{\alpha}_t\} \) and the dataset covariance structure, making it independent of the specific planning task and amenable to offline pre-computation across all timesteps.

\subsubsection{Truncation Timestep Criterion}

To automatically determine the optimal truncation timestep $t_{\mathrm{stop}}$, we adopt the condition number of the covariance matrix as a quantitative indicator, combined with an elbow detection strategy. Specifically, we use the standard $L_2$-norm condition number of the covariance matrix $\boldsymbol{\Sigma}_{0|t}$:
\begin{equation}
\kappa(t) = \frac{\lambda_{\max}(\boldsymbol{\Sigma}_{0|t})}{\lambda_{\min}(\boldsymbol{\Sigma}_{0|t})},
\end{equation}
where $\lambda_{\max}(\cdot)$ and $\lambda_{\min}(\cdot)$ denote the largest and smallest eigenvalues, respectively. This metric quantifies the anisotropy of the distribution: in the deep-noise regime, the Gaussian approximation must cover the yet-undifferentiated multi-mode distribution and is compressed into a low-dimensional subspace, so $\kappa(t)$ approaches the condition number of the data covariance and $\kappa(t) \gg 1$; as reverse denoising proceeds and the process commits to a specific mode, the covariance narrows and becomes approximately isotropic, and $\kappa(t)$ falls back toward one.

Since the update term in the covariance recursion is a diagonal correction, $\boldsymbol{\Sigma}_{0|t}^{-1}$ shares the same eigenvectors as $\boldsymbol{\Sigma}_{\mathrm{data}}^{-1}$. Consequently, the $i$-th eigenvalue can be expressed analytically as
\begin{equation*}
\lambda_i(t) = \frac{1}{1 / \lambda_i^{\mathrm{data}} + (\rho_t - \rho_{t_{\max}})},
\end{equation*}
where $\lambda_i^{\mathrm{data}}$ is the $i$-th eigenvalue of $\boldsymbol{\Sigma}_{\mathrm{data}}$. The minimum and maximum eigenvalues are thus:
\begin{equation*}
\lambda_{\min}(t) = \frac{1}{1 / \lambda_{\min}^{\mathrm{data}} + \Delta \rho_t},\qquad
\lambda_{\max}(t) = \frac{1}{1 / \lambda_{\max}^{\mathrm{data}} + \Delta \rho_t},
\end{equation*}
with $\Delta \rho_t = \rho_t - \rho_{t_{\max}}$. Substituting the above into the condition number definition yields a closed-form expression that depends only on the extreme eigenvalues of the data covariance:
\begin{equation*}
\kappa(t) = \frac{1 / \lambda_{\min}^{\mathrm{data}} + \Delta \rho_t}{1 / \lambda_{\max}^{\mathrm{data}} + \Delta \rho_t}.
\end{equation*}
Since this computation only requires a single eigendecomposition of $\boldsymbol{\Sigma}_{\mathrm{data}}$ performed offline, $\kappa(t)$ for all timesteps can be obtained in constant time $O(1)$.

As the noise level $t$ increases, $\kappa(t)$ grows monotonically toward the condition number of the data covariance (equivalently, $\kappa(t)$ decreases monotonically during the reverse denoising process). The resulting curve typically exhibits a ``flat--steep rise--flat'' pattern, where the elbow point transitioning from flat to steep rise corresponds to the onset of mode separation (the speciation time).
\begin{equation}
t^* = \arg\max_t \Delta^2 \kappa(t).
\end{equation}
The final truncation timestep is set to $t_{\mathrm{stop}} = \min(t^* + \delta, t_{\max})$, where $\delta \in \{0, \ldots, 5\}$ is a small nonnegative offset to ensure the reverse process has fully crossed the mode-selection boundary, stabilizing the Gaussian approximation.

\subsubsection{Mapping from Control Sequence Space to B\'ezier Control-Point Space}

After obtaining the Gaussian distribution \( \mathcal{N}(\boldsymbol{\mu}_U, \boldsymbol{\Sigma}_U) \) at the truncation timestep \( t_{\mathrm{stop}} \), we map it to the low-dimensional B\'ezier control-point space to reduce the dimensionality for the subsequent CMA-ES optimization and make the control sequence smoother. As defined before, the control sequence and B\'ezier control points satisfy the linear transformation \( \mathbf{U} = \mathbf{M}\mathbf{B} \), with pseudo-inverse \( \mathbf{A} = (\mathbf{M}^\top \mathbf{M})^{-1} \mathbf{M}^\top \) satisfying \( \mathbf{B} = \mathbf{A}\mathbf{U} \).

Due to the closure property of Gaussian distributions under linear transformations, the prior distribution in the control-point space is
\begin{equation}
p(\mathbf{B}) = \mathcal{N}(\boldsymbol{\mu}_B, \boldsymbol{\Sigma}_B),
\label{eq:trunc_prior_b}
\end{equation}
where the mean and covariance are given by the linear propagation formulas:
\begin{align}
\boldsymbol{\mu}_B &= \mathbf{A}\boldsymbol{\mu}_U = (\mathbf{M}^\top \mathbf{M})^{-1} \mathbf{M}^\top \boldsymbol{\mu}_U, \label{eq:trunc_mean_prop} \\
\boldsymbol{\Sigma}_B &= \mathbf{A}\boldsymbol{\Sigma}_U \mathbf{A}^\top = (\mathbf{M}^\top \mathbf{M})^{-1} \mathbf{M}^\top \boldsymbol{\Sigma}_U \mathbf{M} (\mathbf{M}^\top \mathbf{M})^{-1}. \label{eq:trunc_cov_prop}
\end{align}
This mapping does not involve explicit matrix inversion of \( \boldsymbol{\Sigma}_U \)—it only appears in a quadratic form with \( \mathbf{M} \)—ensuring excellent numerical stability. Moreover, the matrices \( \mathbf{M} \) and \( \mathbf{A} \) are determined entirely by the planning horizon \( H \) and the number of control points \( P \), allowing them to be precomputed offline.

The choice of B\'ezier curves over alternatives such as B-splines is motivated by three considerations. Firstly, B\'ezier curves are globally supported, yielding a dense mapping matrix \(\mathbf{M}\) that preserves the full covariance structure of the diffusion prior during the linear transformation, whereas B-splines with local support would dilute long-range correlations and potentially distort the anisotropic modal geometry encoded in \(\boldsymbol{\Sigma}_U\). Secondly, B\'ezier curves strictly interpolate the endpoints, which naturally enforces fixed initial and terminal constraints without additional projection or penalty terms. Third, the B\'ezier parameterization introduces no extra hyperparameters beyond the number of control points, avoiding the knot-vector tuning required by B-splines.

The complete construction procedure of the prior distribution is summarized in Algorithm~\ref{alg:trunc_prior_construction}.

\begin{algorithm}[htbp]
\caption{Construction of Diffusion Prior}
\label{alg:trunc_prior_construction}
\begin{algorithmic}[1]
\REQUIRE Training dataset \(\{\mathbf{u}_{\mathrm{data}}^{(i)}\}_{i=1}^M\), noise schedule \(\{\bar{\alpha}_t\}_{t=0}^{t_{\max}}\)
\ENSURE Prior \((\boldsymbol{\mu}_B, \boldsymbol{\Sigma}_B)\), truncation timestep \(t_{\mathrm{stop}}\)

\STATE Compute data mean \(\bar{\mathbf{u}}_{\mathrm{data}}\) and covariance \(\boldsymbol{\Sigma}_{\mathrm{data}}\) (with regularization)
\STATE Perform eigendecomposition of \(\boldsymbol{\Sigma}_{\mathrm{data}}\) to obtain \(\lambda_{\min}^{\mathrm{data}}, \lambda_{\max}^{\mathrm{data}}\)

\STATE \(\boldsymbol{\Sigma}_{0|t_{\max}} \gets \boldsymbol{\Sigma}_{\mathrm{data}}\), \(\rho_t \gets \bar{\alpha}_t / (1 - \bar{\alpha}_t)\)
\FOR{\(t = t_{\max}-1\) \textbf{downto} \(0\)}
    \STATE \(\boldsymbol{\Sigma}_{0|t}^{-1} \gets \boldsymbol{\Sigma}_{0|t+1}^{-1} + (\rho_t - \rho_{t+1}) \mathbf{I}\)
\ENDFOR

\FOR{\(t = 0\) \TO \(t_{\max}\)}
    \STATE \(\kappa(t) \gets \dfrac{1 / \lambda_{\min}^{\mathrm{data}} + \rho_t - \rho_{t_{\max}}}{1 / \lambda_{\max}^{\mathrm{data}} + \rho_t - \rho_{t_{\max}}}\)
\ENDFOR
\STATE \(t^* \gets \arg\max_t \Delta^2\kappa(t)\)\quad (second-difference elbow detection)
\STATE \(t_{\mathrm{stop}} \gets \min(t^* + \delta, t_{\max})\)

\STATE Retrieve \(\boldsymbol{\Sigma}_U \gets \boldsymbol{\Sigma}_{0|t_{\mathrm{stop}}}\)
\STATE Compute \(\boldsymbol{\mu}_U\) via Tweedie's formula at \(t_{\mathrm{stop}}\)
\STATE \(\boldsymbol{\mu}_B \gets \mathbf{A}\boldsymbol{\mu}_U,\quad \boldsymbol{\Sigma}_B \gets \mathbf{A}\boldsymbol{\Sigma}_U \mathbf{A}^\top\)\quad (map to control-point space)
\RETURN \(\boldsymbol{\mu}_B, \boldsymbol{\Sigma}_B, t_{\mathrm{stop}}\)
\end{algorithmic}
\end{algorithm}

Through the above pipeline, the knowledge encapsulated in the expert dataset concerning the synthesis of control sequences is efficiently distilled into a closed-form Gaussian distribution \(\mathcal{N}(\boldsymbol{\mu}_B, \boldsymbol{\Sigma}_B)\). This distribution encapsulates the belief regarding the resulting control sequence at the speciation timestep, thereby furnishing the subsequent optimization procedure with both a principled initial guess and a well-calibrated regularization mechanism.

\subsection{Design of the Gibbs Task Likelihood}
\label{sec:likelihood}

Let \(\mathcal{S}\) represent the degree to which a candidate solution satisfies the task specification, including safety, control smoothness, and goal-reaching requirements. We define the Gibbs task likelihood as
\begin{equation}
p(\mathcal{S} \mid \mathbf{B})
\propto
\exp\left(-\tau \cdot J_{\mathrm{task}}(\mathbf{B})\right),
\label{eq:gibbs_likelihood}
\end{equation}
where \(J_{\mathrm{task}}(\mathbf{B})\) is a scalar energy function that quantifies task violation, and \(\tau>0\) is the inverse temperature parameter controlling the strength of task conditioning. Here, \(p(\mathcal{S}\mid\mathbf{B})\) is treated as a generalized likelihood factor---a nonnegative measure of task compatibility---rather than a strict sampling distribution. For brevity, it is referred to as the task likelihood hereafter. Lower energy corresponds to higher task compatibility, while the posterior normalization is introduced only after this factor is combined with the diffusion prior.
The energy function is constructed as a weighted sum of three components:
\begin{equation}
J_{\mathrm{task}}(\mathbf{B}) = J_{\mathrm{safe}}(\mathbf{B}) + J_{\mathrm{ctrl}}(\mathbf{B}) + J_{\mathrm{goal}}(\mathbf{B}),
\label{eq:task_energy_total}
\end{equation}
where \(J_{\mathrm{safe}}\) penalizes violations of safety constraints, \(J_{\mathrm{ctrl}}\) regularizes the control effort and smoothness, and \(J_{\mathrm{goal}}\) enforces the terminal objective. Each component is evaluated by first reconstructing the control sequence \(\mathbf{U} = \mathbf{M}\mathbf{B}\) via the B\'ezier mapping and then performing forward simulation through the discrete-time dynamics \(\mathbf{x}_{k+1} = f(\mathbf{x}_k, \mathbf{u}_k)\) to obtain the state trajectory \(\mathbf{X}(\mathbf{U})\). The entire evaluation requires only forward rollouts of the dynamics, without access to their analytical form or derivative information, and is therefore applicable to the black-box scenario.

\subsubsection{Safety Energy via Control Barrier Function}

The safety energy quantifies the extent to which a candidate trajectory violates the discrete-time CBF condition. Given the safe set \(\mathcal{C} = \{\mathbf{x} : h(\mathbf{x}) \geq 0\}\), the discrete-time CBF condition requires
\begin{equation*}
h(\mathbf{x}_{k+1}) \geq (1 - \gamma) h(\mathbf{x}_k), \quad 0 < \gamma \leq 1,
\end{equation*}
for all timesteps \(k = 0, \dots, H-1\), where \(\gamma\) controls the decay rate of the safety margin. To penalize violations of this condition, we define the safety energy as
\begin{equation}
J_{\mathrm{safe}}(\mathbf{B}) = \lambda_{\mathrm{safe}} \sum_{k=0}^{H-1} \left[ \max \left(0, (1-\gamma)h(\mathbf{x}_k) - h(\mathbf{x}_{k+1}) \right) \right]^2,
\label{eq:energy_safe_raw}
\end{equation}
where \(\lambda_{\mathrm{safe}} \gg 1\) is a large weighting coefficient chosen so that safety violations dominate the total energy by design, making infeasible trajectories highly unlikely under the posterior distribution.

Since the \(\max(0, \cdot)\) operator introduces non-smoothness that may hinder the optimization process, the Smooth Maximum (Softplus) function is adopted here as a differentiable approximation:
\begin{equation*}
\psi_\eta(z) = \frac{1}{\eta} \log\left(1 + e^{\eta z}\right),
\end{equation*}
where \(\eta > 0\) controls the approximation accuracy. As \(\eta \to \infty\), \(\psi_\eta(z)\) converges uniformly to \(\max(0, z)\). Substituting this into Eq.~(\ref{eq:energy_safe_raw}) yields the smoothed safety energy:
\begin{equation}
J_{\mathrm{safe}}(\mathbf{B}) = \lambda_{\mathrm{safe}} \sum_{k=0}^{H-1} \left[ \psi_\eta \left( (1-\gamma)h(\mathbf{x}_k) - h(\mathbf{x}_{k+1}) \right) \right]^2.
\label{eq:energy_safe_smooth}
\end{equation}
This smooth formulation retains the structure of the original CBF condition as a soft penalty while providing stable energy landscapes for the CMA-ES solver. 

\subsubsection{Auxiliary Energy for Control Performance}

To ensure the generated trajectories are dynamically feasible and comfortable for execution, we introduce two auxiliary regularization terms. The first penalizes large control magnitudes to avoid actuator saturation:
\begin{equation*}
J_{u}(\mathbf{B}) = \lambda_u \sum_{k=0}^{H-1} \|\mathbf{u}_k\|_2^2,
\end{equation*}
where \(\lambda_u > 0\) is the control effort weight. The second penalizes aggressive changes between consecutive control inputs to promote smoothness and reduce mechanical wear:
\begin{equation*}
J_{du}(\mathbf{B}) = \lambda_{du} \sum_{k=0}^{H-2} \|\mathbf{u}_{k+1} - \mathbf{u}_k\|_2^2,
\end{equation*}
where \(\lambda_{du} > 0\) is the smoothness weight. These two terms are combined into the overall control performance energy:
\begin{equation}
J_{\mathrm{ctrl}}(\mathbf{B}) = J_{u}(\mathbf{B}) + J_{du}(\mathbf{B}).
\label{eq:energy_ctrl}
\end{equation}

\subsubsection{Goal Energy}

The goal energy drives the system toward the desired terminal state \(\mathbf{x}_{\mathrm{goal}}\). It is defined as the squared Euclidean distance between the final predicted state and the target:
\begin{equation}
J_{\mathrm{goal}}(\mathbf{B}) = \lambda_g \|\mathbf{x}_H - \mathbf{x}_{\mathrm{goal}}\|_2^2,
\label{eq:energy_goal}
\end{equation}
where \(\lambda_g > 0\) is the goal-reaching weight. If the actual task requires imposing constraints on the positions of other intermediate points along the path, the goal energy can be extended accordingly.

\subsubsection{Energy-to-Likelihood Mapping and Interpretation}

Combining Eqs.~(\ref{eq:energy_safe_smooth}), (\ref{eq:energy_ctrl}), and (\ref{eq:energy_goal}) yields the total energy \(J_{\mathrm{task}}(\mathbf{B})\) in Eq.~(\ref{eq:task_energy_total}). The Gibbs construction in Eq.~(\ref{eq:gibbs_likelihood}) maps this energy into a nonnegative task factor, which serves as the generalized task likelihood within the Gibbs posterior. The inverse temperature \(\tau\) balances the influence of this task likelihood and the diffusion prior. As \(\tau\to0\), the likelihood factor approaches a constant and the Gibbs posterior reduces to the prior. As \(\tau\) increases, the MAP objective is increasingly dominated by the task energy. In practice, \(\tau\) is chosen empirically to promote safety, smoothness, and goal-reaching performance while preserving the mode-consistency inherited from the diffusion prior. Together with the Gaussian prior, this task likelihood defines the Gibbs posterior used in the following section.
\begin{algorithm}[ht]
\caption{Evaluation of the Task Energy}
\label{alg:likelihood}
\begin{algorithmic}[1]
\REQUIRE $\mathbf{B}$ (control points), initial state $\mathbf{x}_{\mathrm{init}}$, system dynamics $f(\cdot)$, CBF $h(\cdot)$, parameters $\gamma,\eta,\lambda_{\mathrm{safe}},\lambda_u,\lambda_{du},\lambda_g,\mathbf{x}_{\mathrm{goal}}$, Bézier mapping matrix $\mathbf{M}$, horizon $H$.
\ENSURE Total energy $J_{\mathrm{task}}(\mathbf{B})$.

\STATE Reconstruct the control sequence $\mathbf{U} = \mathbf{M}\mathbf{B}$ to yield $\mathbf{u}_0,\dots,\mathbf{u}_{H-1}$.
\STATE Initialize state trajectory: $\mathbf{x}_0 \leftarrow \mathbf{x}_{\mathrm{init}}$.
\FOR{$k = 0$ \TO $H-1$}
    \STATE $\mathbf{x}_{k+1} \leftarrow f(\mathbf{x}_k, \mathbf{u}_k)$.
\ENDFOR

\STATE Compute safety energy:
\STATE \quad $\displaystyle J_{\mathrm{safe}} \leftarrow \lambda_{\mathrm{safe}} \sum_{k=0}^{H-1} \left[ \psi_\eta\!\left((1-\gamma)h(\mathbf{x}_k) - h(\mathbf{x}_{k+1})\right) \right]^2$.

\STATE Compute control performance energy:
\STATE \quad $\displaystyle J_{\mathrm{ctrl}} \leftarrow \lambda_u \sum_{k=0}^{H-1} \|\mathbf{u}_k\|^2 \;+\; \lambda_{du} \sum_{k=0}^{H-2} \|\mathbf{u}_{k+1} - \mathbf{u}_k\|^2$.

\STATE Compute goal energy:
\STATE \quad $\displaystyle J_{\mathrm{goal}} \leftarrow \lambda_g \|\mathbf{x}_H - \mathbf{x}_{\mathrm{goal}}\|^2$.

\STATE $J_{\mathrm{task}} \leftarrow J_{\mathrm{safe}} + J_{\mathrm{ctrl}} + J_{\mathrm{goal}}$.
\STATE \textbf{return} $J_{\mathrm{task}}$.
\end{algorithmic}
\end{algorithm}

\subsection{Posterior Computation via CMA-ES}
\label{sec:posterior}

Given the Gaussian prior
$p(\mathbf{B})=\mathcal{N}(\boldsymbol{\mu}_B,\boldsymbol{\Sigma}_B)$
constructed from the truncated diffusion model and the Gibbs
generalized likelihood
$p(\mathcal{S}\mid\mathbf{B})\propto
\exp[-\tau J_{\mathrm{task}}(\mathbf{B})]$, the generalized Bayesian
update yields the normalized Gibbs posterior
\begin{equation}
\begin{aligned}
p(\mathbf{B}\mid\mathcal{S})
&=
\frac{1}{Z_\tau}
\exp\left[-\tau J_{\mathrm{task}}(\mathbf{B})\right]
\mathcal{N}
\left(
\mathbf{B}\mid
\boldsymbol{\mu}_B,\boldsymbol{\Sigma}_B
\right),\\
Z_\tau
&=
\int
\exp\left[-\tau J_{\mathrm{task}}(\widetilde{\mathbf{B}})\right]
\mathcal{N}
\left(
\widetilde{\mathbf{B}}\mid
\boldsymbol{\mu}_B,\boldsymbol{\Sigma}_B
\right)
\mathrm{d}\widetilde{\mathbf{B}}.
\end{aligned}
\label{eq:posterior_def}
\end{equation}

Taking the negative logarithm of this posterior distribution converts the maximum a posteriori estimation into the deterministic minimization problem given in \eqref{eq:map_objective_posterior}.

To solve this optimization problem, we define the fitness function for the CMA-ES solver as the negative log-posterior:
\begin{equation}
f(\mathbf{B}) = \tau \cdot J_{\mathrm{task}}(\mathbf{B}) + \frac{1}{2} (\mathbf{B} - \boldsymbol{\mu}_B)^\top \boldsymbol{\Sigma}_B^{-1} (\mathbf{B} - \boldsymbol{\mu}_B),
\label{eq:fitness_posterior}
\end{equation}
where the first term penalizes task violations and the second term acts as a Mahalanobis distance regularizer that anchors the solution to the diffusion prior. 

The initialization of the CMA-ES solver is directly derived from the truncated diffusion prior. The initial mean of the search distribution is set to the propagated mean of the prior:
\begin{equation*}
\mathbf{m}^{(0)} = \boldsymbol{\mu}_B,
\end{equation*}
and the initial covariance matrix is set to the propagated covariance of the prior:
\begin{equation*}
\mathbf{C}^{(0)} = \boldsymbol{\Sigma}_B.
\end{equation*}
The initial global step-size \(\sigma^{(0)}\) is determined according to the standard CMA-ES heuristic, typically set here as \(\sigma^{(0)} = 1\) in the normalized control-point space, which ensures that the initial sampling range adequately covers the plausible region defined by the prior. All remaining internal parameters, including the population size \(\lambda\), the number of elite parents \(\mu_e\), the learning rates \(c_1\), \(c_\mu\), \(c_c\), \(c_\sigma\), and the damping factor \(d_\sigma\), are automatically computed using established empirical formulas based solely on the problem dimension \(D' = m \times (P+1)\). This automatic parameter configuration renders the solver entirely self-contained and eliminates the need for manual tuning across different tasks.

\begin{algorithm}[htbp]
\caption{Posterior Computation via CMA-ES}
\label{alg:cmaes_posterior}
\begin{algorithmic}[1]
\REQUIRE Prior mean \(\boldsymbol{\mu}_B\), prior covariance \(\boldsymbol{\Sigma}_B\), task energy \(J_{\mathrm{task}}(\cdot)\), inverse temperature \(\tau\), Bézier mapping \(\mathbf{M}\), termination tolerance \(\varepsilon_{\mathrm{tol}}\).
\ENSURE Optimal control sequence \(\mathbf{U}^{*}\), optimal control points \(\mathbf{B}^{*}\).

\STATE Initialize CMA-ES search distribution:
\STATE \quad \(\mathbf{m}^{(0)} \gets \boldsymbol{\mu}_B\), \(\mathbf{C}^{(0)} \gets \boldsymbol{\Sigma}_B\).
\STATE Set \(\sigma^{(0)}\) and all internal CMA-ES parameters based on dimension \(D'\).
\STATE \(g \gets 0\).
\WHILE{not terminated}
    \STATE Sample \(\lambda\) candidates: \(\mathbf{B}_i^{(g)} \sim \mathcal{N}(\mathbf{m}^{(g)}, (\sigma^{(g)})^2 \mathbf{C}^{(g)})\), for \(i=1,\dots,\lambda\).
    \FOR{\(i = 1\) \TO \(\lambda\)}
        \STATE Reconstruct control sequence: \(\mathbf{U}_i^{(g)} \gets \mathbf{M}\mathbf{B}_i^{(g)}\).
        \STATE Simulate trajectory \(\mathbf{X}(\mathbf{U}_i^{(g)})\) via forward dynamics.
        \STATE Compute fitness: \(f_i \gets f(\mathbf{B}_i^{(g)})\) using \eqref{eq:fitness_posterior}.
    \ENDFOR
    \STATE Sort candidates by fitness \(f_i\) (ascending).
    \STATE Update mean \(\mathbf{m}^{(g+1)}\) via weighted recombination of top \(\mu_e\) elites.
    \STATE Update evolution paths \(\mathbf{p}_c\), \(\mathbf{p}_\sigma\).
    \STATE Update covariance matrix \(\mathbf{C}^{(g+1)}\) and step-size \(\sigma^{(g+1)}\).
    \STATE \(g \gets g+1\).
\ENDWHILE
\STATE \(\mathbf{B}^{*} \gets \mathbf{m}^{*}\), \(\mathbf{U}^{*} \gets \mathbf{M}\mathbf{B}^{*}\).
\RETURN \(\mathbf{B}^{*}\), \(\mathbf{U}^{*}\).
\end{algorithmic}
\end{algorithm}

At each generation \(g\), the CMA-ES algorithm performs the following iterative procedures to approximate the posterior mode. First, a population of \(\lambda\) candidate control-point vectors \(\{\mathbf{B}_i^{(g)}\}_{i=1}^{\lambda}\) is sampled from the current search distribution \(\mathcal{N}(\mathbf{m}^{(g)}, (\sigma^{(g)})^2 \mathbf{C}^{(g)})\). Each candidate is then mapped back to the original control sequence via \(\mathbf{U}_i = \mathbf{M}\mathbf{B}_i^{(g)}\), and the corresponding state trajectory is obtained through forward simulation of the discrete-time dynamics. The fitness function \eqref{eq:fitness_posterior} is evaluated for each candidate based on the simulated trajectory. Subsequently, the candidates are ranked according to their fitness values, and the top \(\mu_e\) elite individuals are selected for weighted recombination. The search distribution parameters—namely the mean vector \(\mathbf{m}\), the covariance matrix \(\mathbf{C}\), and the global step-size \(\sigma\)—are then updated following the standard CMA-ES adaptation rules detailed in \cref{sec:cmaes_intro}. This update mechanism leverages the cumulative evolution paths to guide the covariance adaptation, effectively rotating and scaling the search distribution to align with the local landscape of the fitness function.

The optimization process terminates when either the maximum number of generations is reached or the relative improvement of the best fitness value falls below a predefined tolerance threshold \(\varepsilon_{\mathrm{tol}}\). Upon convergence, the final mean vector \(\mathbf{m}^{(*)}\) of the search distribution serves as the optimal control-point solution:
\begin{equation*}
\mathbf{B}^{*} = \mathbf{m}^{(*)}.
\end{equation*}
The corresponding optimal control sequence is then reconstructed through the B\'ezier mapping:
\begin{equation*}
\mathbf{U}^{*} = \mathbf{M}\mathbf{B}^{*}.
\end{equation*}
This sequence is subsequently applied to the system in a receding-horizon fashion or executed as the complete planned trajectory, depending on the specific control architecture.

The entire posterior computation process is summarized in Algorithm~\ref{alg:cmaes_posterior}, which illustrates how the diffusion prior, the task likelihood, and the CMA-ES solver are seamlessly integrated within the proposed Bayesian inference framework.

With the posterior mode obtained, the planned control sequence is ready for execution.

\section{Theoretical Analysis}

This section analyzes the theoretical properties of TDP-SafeMAP in three parts: V-A derives a local MAP perturbation bound and a sufficient condition for discrete-time CBF safety preservation; V-B places the truncated prior and the CMA-ES posterior search within a unified statistical view, interpreting the truncation through Bayesian risk and forward-KL projection, and connecting the solver's rank-based updates to natural-gradient ascent on the Gaussian manifold; V-C establishes the exact equivalence between the original and prior-whitened MAP problems. All proofs are deferred to the appendix.

\subsection{Approximation Error Bound and Safety Preservation}

This subsection analyzes how diffusion-prior approximation and finite-accuracy posterior optimization affect the resulting trajectory, and establishes a local perturbation bound for the MAP solution and a sufficient condition for preserving discrete-time CBF safety.

Let \(p^{\circ}(\mathbf{B})\) denote the ideal within-mode diffusion prior in the B\'ezier control-point space. To distinguish it from the Gaussian prior used by TDP-SafeMAP, the latter is denoted by \(\widehat{p}(\mathbf{B})\) in this section. Define
\begin{align*}
F^{\circ}(\mathbf{B})
&=
\tau J_{\mathrm{task}}(\mathbf{B})
-
\log p^{\circ}(\mathbf{B}),\\
\widehat{F}(\mathbf{B})
&=
\tau J_{\mathrm{task}}(\mathbf{B})
-
\log \widehat{p}(\mathbf{B}),
\end{align*}
and the prior-induced objective perturbation
\begin{equation*}
\Delta_p(\mathbf{B})
=
-\log\widehat{p}(\mathbf{B})
+
\log p^{\circ}(\mathbf{B}).
\end{equation*}
Since additive constants do not affect the MAP solution, the relevant prior error over a local mode region
\(\mathcal{R}_{\mathrm{mode}}\) is
\begin{equation}
\omega_p
=
\sup_{\mathbf{B}\in\mathcal{R}_{\mathrm{mode}}}
\Delta_p(\mathbf{B})
-
\inf_{\mathbf{B}\in\mathcal{R}_{\mathrm{mode}}}
\Delta_p(\mathbf{B}).
\end{equation}

\begin{assumption}[Local stability and regularity]
Let
\begin{equation}
\mathbf{B}^{\circ}
\in
\arg\min_{\mathbf{B}\in\mathcal{R}_{\mathrm{mode}}}
F^{\circ}(\mathbf{B})
\end{equation}
be the ideal local MAP solution. The reference and implemented priors are strictly positive on the compact region
\(\mathcal{R}_{\mathrm{mode}}\), and the following conditions hold:
\begin{align}
F^{\circ}(\mathbf{B})
&\geq
F^{\circ}(\mathbf{B}^{\circ})
+
\frac{\vartheta_{\mathrm{qg}}}{2}
\left\|
\mathbf{B}-\mathbf{B}^{\circ}
\right\|^2,\\
\widehat{F}(\widehat{\mathbf{B}})
&\leq
\inf_{\mathbf{B}\in\mathcal{R}_{\mathrm{mode}}}
\widehat{F}(\mathbf{B})
+
\zeta_{\mathrm{opt}},
\end{align}
for all
\(\mathbf{B}\in\mathcal{R}_{\mathrm{mode}}\), where
\(\vartheta_{\mathrm{qg}}>0\) and
\(\zeta_{\mathrm{opt}}\geq0\). Moreover, on a compact neighborhood containing the relevant rollouts,
\begin{align*}
\left\|
f(\mathbf{x},\mathbf{u})
-
f(\mathbf{x}',\mathbf{u}')
\right\|
&\leq
L_x\left\|\mathbf{x}-\mathbf{x}'\right\|
+
L_u\left\|\mathbf{u}-\mathbf{u}'\right\|,\\
\left|
h(\mathbf{x})-h(\mathbf{x}')
\right|
&\leq
L_h\left\|\mathbf{x}-\mathbf{x}'\right\|.
\end{align*}
\end{assumption}

The quadratic-growth condition is local and does not require differentiability of the task energy, the black-box dynamics, or the CBF.

\begin{theorem}[Error bound for the approximate MAP solution]
Under the preceding assumption, the CMA-ES solution satisfies
\begin{equation}
\left\|
\widehat{\mathbf{B}}
-
\mathbf{B}^{\circ}
\right\|
\leq
\sqrt{
\frac{
2\left(\omega_p+\zeta_{\mathrm{opt}}\right)
}{
\vartheta_{\mathrm{qg}}
}
}.
\end{equation}
\end{theorem}

\begin{IEEEproof}
See Appendix~\ref{app:proof_map_error}.
\end{IEEEproof}

To make this result explicit, define the ideal and implemented Gaussian
priors and their parameter discrepancies by
\begin{equation*}
\begin{aligned}
p_G^{\circ}
&=
\mathcal{N}
\left(
\boldsymbol{\mu}_B^{\circ},
\boldsymbol{\Sigma}_B^{\circ}
\right),
&
\widehat{p}
&=
\mathcal{N}
\left(
\boldsymbol{\mu}_B,
\boldsymbol{\Sigma}_B
\right),
\\
\upsilon_\mu
&=
\left\|
\boldsymbol{\mu}_B-\boldsymbol{\mu}_B^{\circ}
\right\|,
&
\upsilon_\Sigma
&=
\left\|
\boldsymbol{\Sigma}_B-\boldsymbol{\Sigma}_B^{\circ}
\right\|.
\end{aligned}
\end{equation*}
Furthermore, let
\begin{equation*}
s_{\min}^{\circ}
=
\lambda_{\min}
\left(
\boldsymbol{\Sigma}_B^{\circ}
\right),
\quad
\mathcal{R}_{\mathrm{mode}}
\subseteq
\left\{
\mathbf{B}:
\left\|
\mathbf{B}-\boldsymbol{\mu}_B^{\circ}
\right\|
\leq
R_{\mathrm{mode}}
\right\}.
\end{equation*}
If \(p^{\circ}\) is not exactly Gaussian, define its Gaussianization-error oscillation as
\begin{align}
\omega_G
=&\,
\sup_{\mathbf{B}\in\mathcal{R}_{\mathrm{mode}}}
\left[
-\log p_G^{\circ}(\mathbf{B})
+
\log p^{\circ}(\mathbf{B})
\right]
\nonumber\\
&-
\inf_{\mathbf{B}\in\mathcal{R}_{\mathrm{mode}}}
\left[
-\log p_G^{\circ}(\mathbf{B})
+
\log p^{\circ}(\mathbf{B})
\right].
\end{align}
When \(p^{\circ}=p_G^{\circ}\), one has \(\omega_G=0\).

For the B\'ezier mapping
\(\mathbf{u}_k=\mathbf{M}_k\mathbf{B}\), define
\begin{equation}
G_0=0,
\qquad
G_k
=
L_u
\sum_{j=0}^{k-1}
L_x^{\,k-1-j}
\left\|\mathbf{M}_j\right\|,
\end{equation}
and
\begin{equation}
K_{\mathrm{CBF}}
=
\max_{0\leq k<H}
L_h
\left[
G_{k+1}
+
(1-\gamma)G_k
\right].
\end{equation}

The preceding prior-related quantities provide an explicit bound on
\(\omega_p\), while \(K_{\mathrm{CBF}}\) quantifies the propagation of
control-point errors to the CBF residuals; together with the preceding
MAP error bound, they yield the following explicit sufficient condition
for safety preservation.

\begin{theorem}[Explicit condition for CBF safety preservation]
Suppose that
\begin{equation}
h(\mathbf{x}_0)\geq0,
\qquad
\varrho_{\mathrm{CBF}}
=
\min_{0\leq k<H}
r_k(\mathbf{B}^{\circ})
>0,
\end{equation}
where
\begin{equation}
r_k(\mathbf{B})
=
h\bigl(\mathbf{x}_{k+1}(\mathbf{B})\bigr)
-
(1-\gamma)
h\bigl(\mathbf{x}_k(\mathbf{B})\bigr).
\end{equation}
If
\(\upsilon_\Sigma<s_{\min}^{\circ}\),
\(K_{\mathrm{CBF}}>0\), and
\begin{equation}
\begin{aligned}
\omega_G
&+
R_{\mathrm{mode}}^2
\frac{
\upsilon_\Sigma
}{
s_{\min}^{\circ}
\left(
s_{\min}^{\circ}-\upsilon_\Sigma
\right)
}
+
2R_{\mathrm{mode}}
\frac{
\upsilon_\mu
}{
s_{\min}^{\circ}-\upsilon_\Sigma
}
+
\zeta_{\mathrm{opt}}
\\
&\leq
\frac{
\vartheta_{\mathrm{qg}}
\varrho_{\mathrm{CBF}}^2
}{
2K_{\mathrm{CBF}}^2
},
\end{aligned}
\end{equation}
then the computed trajectory satisfies
\begin{equation}
r_k(\widehat{\mathbf{B}})\geq0,
\qquad
h(\widehat{\mathbf{x}}_k)\geq0,
\end{equation}
for all \(k\) in the planning horizon.
\end{theorem}

\begin{IEEEproof}
See Appendix~\ref{app:proof_cbf_safety}.
\end{IEEEproof}

The condition separates the effects of Gaussianization error, Gaussian parameter error, CMA-ES suboptimality, local MAP curvature, and rollout sensitivity. It is a robustness result: a strict nominal CBF margin is preserved when the combined approximation and optimization errors are sufficiently small.

\subsection{A Unified Statistical Interpretation of TDP-SafeMAP}
\label{sec:cma_compatibility}

The overall pipeline of TDP-SafeMAP can be uniformly
 interpreted from a statistical perspective. 
  We provide a derivation and analysis of this pipeline below.

Let
\begin{equation*}
s_t
=
\rho_t-\rho_{t_{\max}}
\geq 0.
\end{equation*}
Under the Gaussian second-order closure in
Eq.~\eqref{eq:trunc_cov_recursion}, the covariance recursion satisfies
\begin{equation*}
\boldsymbol{\Sigma}_{0|t}
=
\left(
\boldsymbol{\Sigma}_{\mathrm{data}}^{-1}
+s_t\mathbf{I}
\right)^{-1}.
\end{equation*}
Therefore, its variance in the \(i\)-th data-covariance eigendirection is
\begin{equation}
v_i(t)
=
\frac{\lambda_i^{\mathrm{data}}}
{1+s_t\lambda_i^{\mathrm{data}}},
\label{eq:directional_bayes_risk}
\end{equation}
which equals the Bayes risk of the corresponding surrogate Gaussian
estimation problem. In particular,
\begin{equation}
\kappa(t)
=
1+c_{\Sigma}v_{\max}(t),
\qquad
c_{\Sigma}
=
\frac{1}{\lambda_{\min}^{\mathrm{data}}}
-
\frac{1}{\lambda_{\max}^{\mathrm{data}}}.
\label{eq:kappa_bayes_risk_identity}
\end{equation}
Hence,
\begin{equation}
\arg\max_t\Delta^2\kappa(t)
=
\arg\max_t\Delta^2v_{\max}(t).
\label{eq:bayes_risk_truncation}
\end{equation}
The second-difference criterion thus detects the largest change in the
per-step reduction of the principal-direction Bayes risk. When the
between-mode covariance is concentrated mainly in this direction, the same
quantity provides a second-order indicator of the transition from unresolved
mode uncertainty to within-mode refinement. This interpretation does not
require the global diffusion distribution to be Gaussian or a finite Gaussian
mixture.

At \(t_{\mathrm{stop}}\), let
\(\pi_{\mathrm{stop}}^U\) denote the conditional clean-control distribution
associated with the realized reverse-diffusion branch. Its forward-KL
projection onto the Gaussian family is
\begin{equation}
\widehat p_U
=
\arg\min_{q\in\mathcal{G}_D}
D_{\mathrm{KL}}
\left(
\pi_{\mathrm{stop}}^U\|q
\right)
=
\mathcal{N}
\left(
\boldsymbol{\mu}_U,
\boldsymbol{\Sigma}_U
\right),
\label{eq:prior_gaussian_projection}
\end{equation}
where the implemented moments are obtained from Tweedie's formula and the
covariance recursion. Since the B\'ezier transformation
\(\mathbf{B}=\mathbf{A}\mathbf{U}\) is linear, the resulting prior is
\begin{equation}
p_{\mathrm{pri}}(\mathbf{B})
=
\widehat p(\mathbf{B})
=
\mathcal{N}
\left(
\boldsymbol{\mu}_B,
\boldsymbol{\Sigma}_B
\right),
\label{eq:projected_prior_theory}
\end{equation}
where
\(\boldsymbol{\mu}_B=\mathbf{A}\boldsymbol{\mu}_U\) and
\(\boldsymbol{\Sigma}_B
=\mathbf{A}\boldsymbol{\Sigma}_U\mathbf{A}^{\top}\).

The task energy then produces the Gibbs posterior
\begin{equation}
p_\tau(\mathbf{B}\mid\mathcal{S})
=
\frac{1}{Z_\tau}
\exp[-\tau J_{\mathrm{task}}(\mathbf{B})]
p_{\mathrm{pri}}(\mathbf{B}),
\label{eq:gibbs_posterior_theory}
\end{equation}
which is equivalently the minimum-information update
\begin{equation}
p_\tau
=
\arg\min_{\nu}
\left\{
\mathbb{E}_{\nu}[J_{\mathrm{task}}]
+
\frac{1}{\tau}
D_{\mathrm{KL}}(\nu\|p_{\mathrm{pri}})
\right\}
\label{eq:gibbs_variational}
\end{equation}
over distributions absolutely continuous with respect to
\(p_{\mathrm{pri}}\)~\cite{11039694}. For the Gaussian prior,
\begin{equation*}
-\log p_\tau(\mathbf{B}\mid\mathcal{S})
=
f(\mathbf{B})+c_\tau,
\end{equation*}
where \(c_\tau\) is independent of \(\mathbf{B}\). Consequently,
\begin{equation}
f(\mathbf{B}_i)<f(\mathbf{B}_j)
\quad\Longleftrightarrow\quad
p_\tau(\mathbf{B}_i\mid\mathcal{S})
>
p_\tau(\mathbf{B}_j\mid\mathcal{S}).
\label{eq:fitness_posterior_order}
\end{equation}
Thus, the unknown normalizing constant \(Z_\tau\) is irrelevant to CMA-ES:
ranking candidates by fitness is exactly equivalent to ranking them by
posterior density.

The search distribution used by Algorithm~\ref{alg:cmaes_posterior} is
\begin{equation}
q_g^{\mathrm{CMA}}(\mathbf{B})
=
\mathcal{N}
\left(
\mathbf{B}\mid
\mathbf{m}^{(g)},
(\sigma^{(g)})^2\mathbf{C}^{(g)}
\right),
\end{equation}
with initialization
\begin{equation}
q_0^{\mathrm{CMA}}(\mathbf{B})
=
\mathcal{N}
\left(
\mathbf{B}\mid
\boldsymbol{\mu}_B,
(\sigma^{(0)})^2\boldsymbol{\Sigma}_B
\right).
\label{eq:actual_cma_initial_law}
\end{equation}
Hence, the search is initialized at the prior mean and inherits the prior
eigendirections and condition number, while \(\sigma^{(0)}\) controls only
the overall search radius. A formal statement is given in Section V-C, 
with its proof provided in Appendix C.

For the complete CMA-ES update used in this work, the mean and rank-\(\mu\)
terms have the established Gaussian IGO interpretation as stochastic natural
gradient updates~\cite{akimoto2012theoretical,ollivier2017information}. The
evolution-path rank-one term has an established MAP-IGO interpretation based
on a generation-dependent normal-inverse-Wishart prior over the Gaussian
search parameters; the standard CMA-ES mean update is recovered in the
non-informative mean-prior limit while the rank-one covariance term is
retained~\cite{hamano2024natural}. Cumulative step-size adaptation further
compares the length of the accumulated whitened search path with its
isotropic Gaussian reference value~\cite{hansen2001completely,HANSEN20115755}.

The distributions involved nevertheless have distinct meanings. The
diffusion prior \(p_{\mathrm{pri}}\) describes uncertainty about the planning
variable, \(p_\tau\) is the generalized Bayesian target, and
\(q_g^{\mathrm{CMA}}\) is an adaptive search distribution. In general,
\begin{equation}
q_g^{\mathrm{CMA}}(\mathbf{B})
\neq
p_\tau(\mathbf{B}\mid\mathcal{S}).
\end{equation}
Accordingly, CMA-ES is used to locate the mode of the Gibbs posterior, rather
than to approximate the complete posterior distribution. Detailed
derivations and the precise distinction between the diffusion prior and the
algorithmic prior used in the rank-one interpretation are provided in the
appendix.

\subsection{Exact MAP Equivalence under Prior Whitening}
\label{sec:prior_whitening}

Assume that the propagated diffusion-prior covariance is positive definite and let
\begin{equation*}
\boldsymbol{\Sigma}_B
=
\mathbf{L}_B\mathbf{L}_B^\top,
\qquad
\mathbf{b}
=
\mathbf{L}_B^{-1}
\left(
\mathbf{B}-\boldsymbol{\mu}_B
\right).
\end{equation*}
The inverse map is
\(\mathcal{T}_B(\mathbf{b})
=\boldsymbol{\mu}_B+\mathbf{L}_B\mathbf{b}\).
Define the whitened objective
\begin{equation}
f_{\mathrm{w}}(\mathbf{b})
=
\tau
J_{\mathrm{task}}
\left(
\mathcal{T}_B(\mathbf{b})
\right)
+
\frac{1}{2}\|\mathbf{b}\|^2.
\label{eq:whitened_map_objective}
\end{equation}

\begin{theorem}[Exact equivalence of prior-whitened MAP]
\label{thm:whitened_map_equivalence}
The MAP problem in Eq.~\eqref{eq:fitness_posterior} and its whitened counterpart satisfy
\begin{equation}
\min_{\mathbf{B}}f(\mathbf{B})
=
\min_{\mathbf{b}}f_{\mathrm{w}}(\mathbf{b}),
\qquad
\mathbf{B}^{*}
=
\boldsymbol{\mu}_B+\mathbf{L}_B\mathbf{b}^{*}.
\label{eq:whitened_map_equivalence}
\end{equation}
Because \(\mathcal{T}_B\) is an invertible affine map with a constant Jacobian determinant, the corresponding posterior modes are also in one-to-one correspondence.
\end{theorem}

\begin{IEEEproof}
See Appendix~\ref{app:proof_whitened_map}.
\end{IEEEproof}

Under the prior-matching initialization used by TDP-SafeMAP,
\begin{equation}
q_0^{\mathrm{CMA}}(\mathbf{B})
=
\mathcal{N}
\left(
\mathbf{B}\mid
\boldsymbol{\mu}_B,\boldsymbol{\Sigma}_B
\right),
\end{equation}
the whitened initial search law is exactly
\begin{equation}
q_{0,\mathrm{w}}^{\mathrm{CMA}}(\mathbf{b})
=
\mathcal{N}(\mathbf{b}\mid\mathbf{0},\mathbf{I}).
\label{eq:whitened_initial_search}
\end{equation}
Equivalently,
\((\sigma^{(0)})^2\mathbf{C}^{(0)}
=\boldsymbol{\Sigma}_B\)
in the original coordinates and
\((\sigma^{(0)})^2\mathbf{C}_{\mathrm{w}}^{(0)}
=\mathbf{I}\)
in the whitened coordinates. Thus, \(\sigma^{(0)}\) is only an internal factorization parameter when the total initial covariance is fixed to the diffusion prior.

For any two Gaussian search laws, the invertible affine map \(\mathcal{T}_B\) preserves both their KL divergence and the expected rank utility induced by the transformed objective. Consequently, the Gaussian natural-gradient core of CMA-ES is equivariant under \(\mathcal{T}_B\):
\begin{equation}
\text{prior-matched search in \(\mathbf{B}\)}
\quad\equiv\quad
\text{isotropic search in \(\mathbf{b}\)}.
\label{eq:whitened_search_equivalence}
\end{equation}
Hence, \(\boldsymbol{\Sigma}_B\) is not merely a warm start; it defines an exact prior-induced preconditioner.

If \(J_{\mathrm{task}}\) is twice differentiable near \(\mathbf{B}^{*}\), then
\begin{equation}
\nabla_{\mathbf{b}}^2 f_{\mathrm{w}}
=
\mathbf{I}
+
\tau
\mathbf{L}_B^\top
\nabla_{\mathbf{B}}^2J_{\mathrm{task}}
\mathbf{L}_B.
\label{eq:whitened_hessian}
\end{equation}
Thus, prior-induced anisotropy is absorbed into the identity, and CMA-ES only needs to adapt to the residual curvature introduced by the task energy. If
\(\mathbf{L}_B^\top
\nabla_{\mathbf{B}}^2J_{\mathrm{task}}
\mathbf{L}_B\succeq0\),
then
\begin{equation}
\kappa
\left(
\nabla_{\mathbf{b}}^2f_{\mathrm{w}}
\right)
=
\frac{
1+\tau\lambda_{\max}
\left(
\mathbf{L}_B^\top
\nabla_{\mathbf{B}}^2J_{\mathrm{task}}
\mathbf{L}_B
\right)
}{
1+\tau\lambda_{\min}
\left(
\mathbf{L}_B^\top
\nabla_{\mathbf{B}}^2J_{\mathrm{task}}
\mathbf{L}_B
\right)
}.
\label{eq:whitened_condition_number}
\end{equation}

The analysis above uses the positive-definite covariance produced by the regularized prior-construction procedure. If an unregularized covariance is only positive semidefinite, exact whitening must instead be performed on its effective support subspace. Adding a diagonal term produces a full-dimensional approximation rather than an exact reparameterization of the singular prior.

\section{Experiments}

\subsection{Experimental Setup}
\label{sec:experiments_setup}

We evaluate TDP-SafeMAP on three navigation environments and one locomotion task from the D4RL benchmark. For navigation, we consider three maze2d environments with varying maze sizes and obstacle layouts, in which a planar point-mass robot navigates toward a goal location inside walled mazes. For locomotion, we adopt a single-legged hopping task, where the agent must move forward while maintaining balance. In all tasks, the transition dynamics are accessible only through black-box forward simulation without gradients, which is consistent with the assumptions of the proposed formulation.

The safety constraints are task-specific. In maze2d, the maze walls are
approximated by fourth-order superellipses. Concretely, each wall segment is
described by
\begin{equation*}
\left|\frac{x-x_0}{a}\right|^4 + \left|\frac{y-y_0}{b}\right|^4 = 1,
\end{equation*}
where $(x_0, y_0)$ is the wall center and $a, b > 0$ are its semi-axes. The
safety function $h(\mathbf{x})$ at each step is the signed distance from the
robot position to the nearest obstacle boundary, which is positive outside the
obstacles and negative inside. In hopper, a single safety constraint is imposed, requiring that the highest
body point not exceed a ceiling of $1.8$\,m above the ground.
This defines the safe set $\mathcal{C}=\{\mathbf{x}: h(\mathbf{x})\geq 0\}$ used
by the CBF-based likelihood in the proposed framework.

We compare against Diffuser~\cite{janner2022diffuser} and SafeDiffuser~\cite{xiao2025safe}. On the maze tasks we include all three
SafeDiffuser variants, namely TVS, RoS and ReS, whereas on the hopper task we use
the RoS variant, which is the most suitable for the red-line-type constraint
considered here because there is no local trap problem in robot locomotion. 
To ensure fairness, all planners share the same pretrained
diffusion backbone, and no method is retrained for the evaluation tasks. 
For the diffusion model applied to the maze2d environment, we follow~\cite{NEURIPS2023_893a5db6} to determine the optimizer, number of training steps and batch size. For the locomotion environments,
 we directly use the official pretrained diffusion models released by Diffuser.

All results are averaged over five random seeds and reported as
mean~$\pm$~std. The CBF decay parameter is $\gamma=0.5$ in both settings, and the
safety-related energy weights are set much larger than the other terms so that
 safety violations dominate the total energy.

\subsection{Maze Navigation on Maze2d}
\label{sec:experiments_maze}

To evaluate the safety and planning performance of TDP-SafeMAP, we compare it
against Diffuser and the three SafeDiffuser variants (TVS, RoS and ReS) on the
three maze2d tasks (Umaze, Medium and Large). The following metrics are adopted: the \emph{minimum safety
distance}, i.e., the minimum signed distance of the trajectory to the obstacle
boundaries; the \emph{path smoothness} measured as the standard deviation of
the heading-angle changes, i.e., $\mathrm{std}_k\big(\operatorname{wrap}(\theta_{k+1}-\theta_k)\big)$ with 
$\theta_k=\operatorname{atan2}(y_{k+1}-y_k,\,x_{k+1}-x_k)$ and $\operatorname{wrap}(\cdot)$ mapping angles to $[-\pi,\pi]$;
 the \emph{total path length}, i.e., the sum of the
Euclidean distances between consecutive trajectory states; and the \emph{average
goal distance}, i.e., the mean Euclidean distance from the trajectory states to
the goal. All metrics are computed by the same evaluation pipeline for every
method.

\begin{table*}[t]
\caption{Quantitative comparison on the maze2d tasks. The best result in each
row group is highlighted in bold. `Prior' denotes the TDP-SafeMAP variant that
uses only the mean of the truncated diffusion prior without posterior
refinement.}
\label{tab:maze2d}
\centering
\small
\begin{tabular}{lcccc}
\hline
Method & Minimum safety distance (m) & Path smoothness (rad) & Total path length (m) & Average goal distance (m)\\
\hline
\multicolumn{5}{c}{\emph{maze2d-umaze-v1}}\\
Diffuser & $0.3370 \pm 0.0370$ & $0.8591 \pm 0.1601$ & $4.2447 \pm 0.2005$ & $1.4592 \pm 0.1093$\\
SafeDiffuser-RoS & $0.3493 \pm 0.0498$ & $0.6181 \pm 0.0867$ & $3.7434 \pm 0.5894$ & $\bm{1.4225 \pm 0.1544}$\\
SafeDiffuser-ReS & $0.3763 \pm 0.0525$ & $0.6117 \pm 0.0712$ & $3.7556 \pm 0.5955$ & $1.4611 \pm 0.1624$\\
SafeDiffuser-TVS & $0.3512 \pm 0.0499$ & $0.6175 \pm 0.0858$ & $3.7441 \pm 0.5899$ & $1.4234 \pm 0.1551$\\
TDP-SafeMAP (prior) & $0.3754 \pm 0.0861$ & $0.2146 \pm 0.2825$ & $4.0964 \pm 0.3368$ & $1.4760 \pm 0.1145$\\
TDP-SafeMAP & $\bm{0.4479 \pm 0.0746}$ & $\bm{0.0837 \pm 0.0781}$ & $\bm{3.5978 \pm 0.5290}$ & $1.5275 \pm 0.1341$\\
\multicolumn{5}{c}{\emph{maze2d-medium-v1}}\\
Diffuser & $0.1952 \pm 0.1155$ & $1.0135 \pm 0.1663$ & $9.0967 \pm 0.1881$ & $\bm{2.6395 \pm 0.6851}$\\
SafeDiffuser-RoS & $-0.1326 \pm 0.0373$ & $1.0582 \pm 0.1292$ & $16.0125 \pm 3.9745$ & $2.7213 \pm 0.7278$\\
SafeDiffuser-ReS & $-0.0835 \pm 0.0273$ & $0.8608 \pm 0.1195$ & $10.3124 \pm 0.8364$ & $2.6893 \pm 0.6801$\\
SafeDiffuser-TVS & $-0.0455 \pm 0.0029$ & $0.8178 \pm 0.1414$ & $9.5360 \pm 0.8288$ & $2.7253 \pm 0.7382$\\
TDP-SafeMAP (prior) & $-0.0824 \pm 0.1908$ & $0.2150 \pm 0.1002$ & $\bm{6.4314 \pm 0.8958}$ & $3.3536 \pm 0.7407$\\
TDP-SafeMAP & $\bm{0.2817 \pm 0.0450}$ & $\bm{0.1084 \pm 0.0406}$ & $6.8423 \pm 1.0704$ & $3.0867 \pm 0.5466$\\
\multicolumn{5}{c}{\emph{maze2d-large-v1}}\\
Diffuser & $0.1692 \pm 0.0453$ & $1.2124 \pm 0.0595$ & $15.7248 \pm 0.7874$ & $3.2262 \pm 1.1265$\\
SafeDiffuser-RoS & $-0.1594 \pm 0.0738$ & $1.3415 \pm 0.6553$ & $30.6562 \pm 16.5695$ & $2.9995 \pm 0.9983$\\
SafeDiffuser-ReS & $0.0716 \pm 0.3163$ & $1.0760 \pm 0.4991$ & $20.4730 \pm 12.5162$ & $2.9622 \pm 1.1368$\\
SafeDiffuser-TVS & $-0.0253 \pm 0.1000$ & $1.0406 \pm 0.1611$ & $16.7758 \pm 4.7915$ & $\bm{2.9570 \pm 1.2280}$\\
TDP-SafeMAP (prior) & $-0.2585 \pm 0.5511 $ & $0.2362 \pm 0.0433$ & $12.0001 \pm 1.3581$ & $4.7818 \pm 1.5800$\\
TDP-SafeMAP & $\bm{0.2229 \pm 0.0878}$ & $\bm{0.1327 \pm 0.0373}$ & $\bm{9.8000 \pm 2.3403}$ & $3.0867 \pm 1.2777$\\
\hline
\end{tabular}
\end{table*}

The quantitative results are summarized in Table~\ref{tab:maze2d}. It
 can be observed that TDP-SafeMAP achieves collision-free trajectories 
 across all tasks and attains the best safety performance, whereas
  the other methods violate the constraints in some tasks. A representative
result on the Umaze task is shown in
Fig.~\ref{fig:maze_umaze}. The planned path selects a task-appropriate
route from the start to the goal, stays well clear of the walls, and
remains smooth and close to the shortest feasible path, which is
consistent with the quantitative gains reported in Table~\ref{tab:maze2d}.

\begin{figure}[!t]
\centering
\includegraphics[width=0.4\columnwidth]{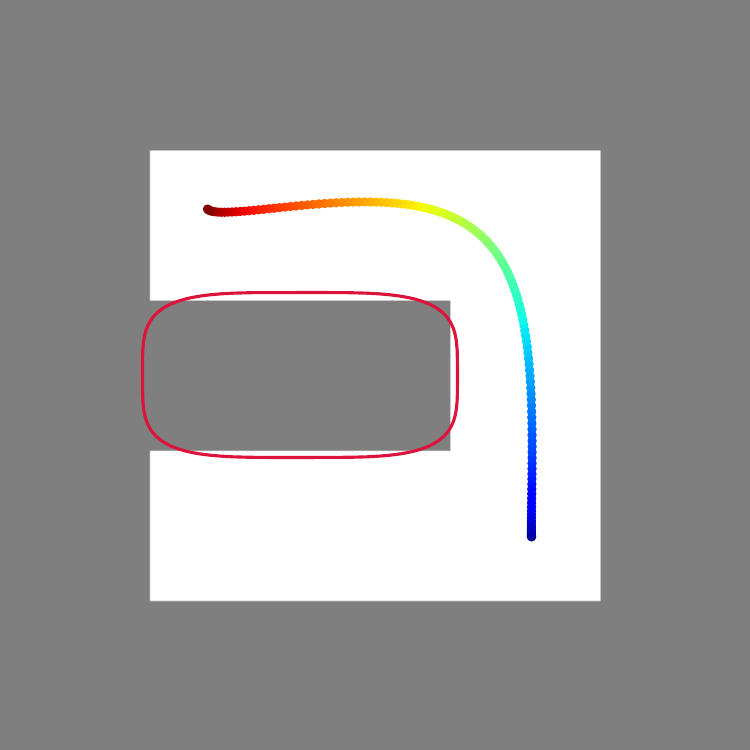}
\caption{Result on the \texttt{maze2d-umaze-v1} task. TDP-SafeMAP
selects a task-appropriate route, keeps a clear margin from the walls, and
produces a smooth path close to the shortest feasible one.}
\label{fig:maze_umaze}
\end{figure}

These gains can be attributed to 
  the individual components of the proposed
framework. The CBF-based energy term directly penalizes violations of the
forward-invariance condition at every step, which explains the consistently
positive minimum safety distance. The Bézier parameterization removes redundant high-frequency 
degrees of freedom in the decision space, and forward-simulating the control sequence through 
the integrator-type dynamics further attenuates high-frequency control components;
 together, these yield a smaller standard deviation of heading-angle changes, i.e., 
 smoother trajectories. Moreover, since the whole pipeline relies only
on black-box forward simulation, the CMA-ES solver requires no gradient of the
task energy, which avoids the bias introduced by differentiating through a
surrogate.

To demonstrate the differences among the final result obtained by TDP-SafeMAP, 
the result of its prior distribution mean and the baseline algorithm, 
we compare them on the Large task in Fig.~\ref{fig:maze_large}.
The unconstrained Diffuser trajectory may approach the walls due to the lack 
of safety constraints and exhibits visible turning jitter. The truncated prior locates the required corridors, 
and the subsequent energy-guided posterior optimization 
yields a control sequence that generates a smooth and safe path.

\begin{figure}[!t]
\centering
\begin{subfigure}[b]{0.31\columnwidth}
\includegraphics[width=\linewidth]{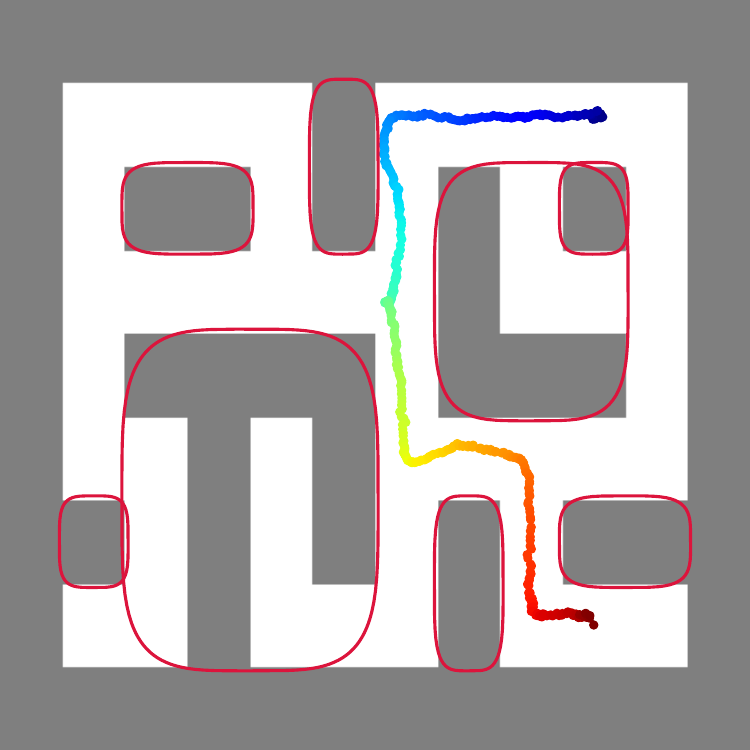}
\caption{Diffuser}
\label{fig:maze_large_diffuser}
\end{subfigure}
\hfill
\begin{subfigure}[b]{0.31\columnwidth}
\includegraphics[width=\linewidth]{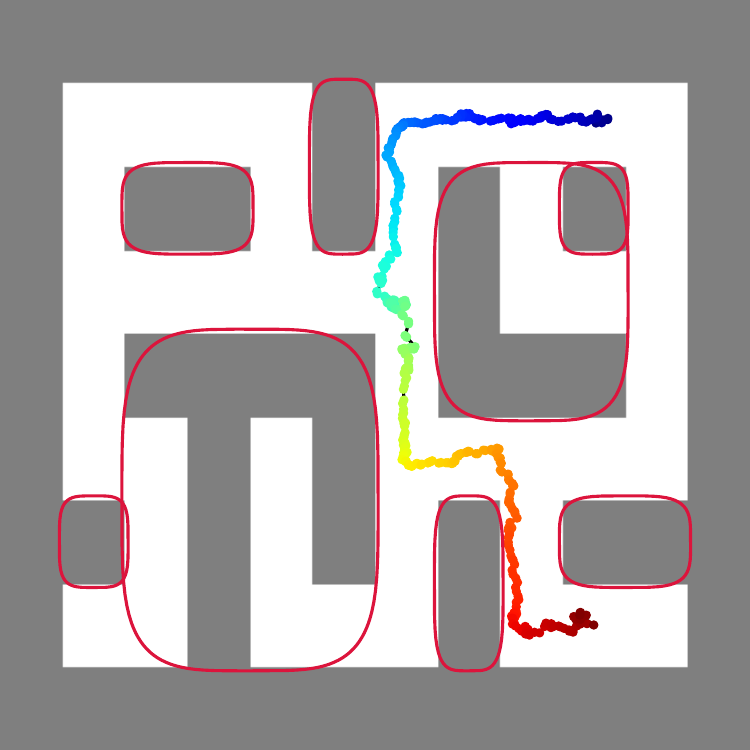}
\caption{Truncated prior}
\label{fig:maze_large_prior}
\end{subfigure}
\hfill
\begin{subfigure}[b]{0.31\columnwidth}
\includegraphics[width=\linewidth]{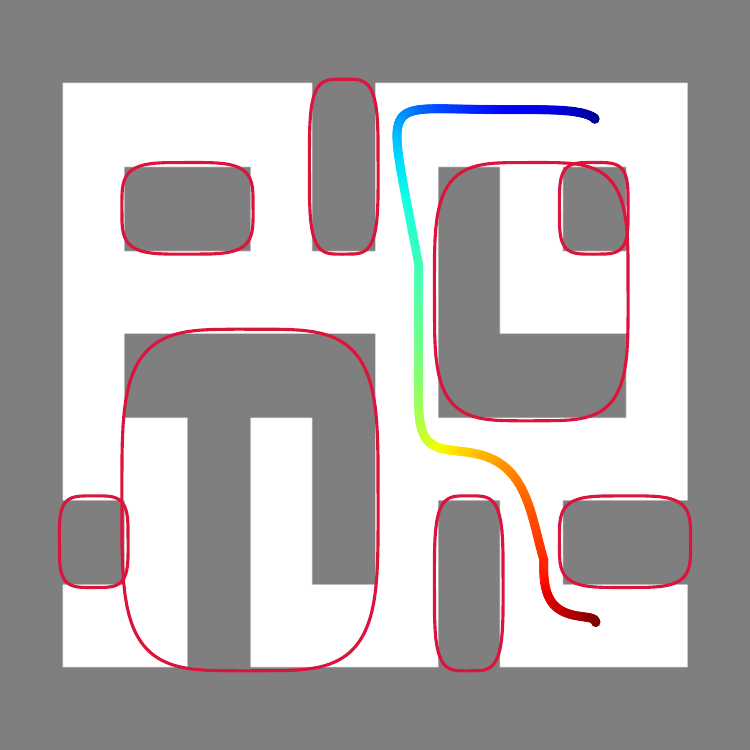}
\caption{TDP-SafeMAP}
\label{fig:maze_large_cmap}
\end{subfigure}
\caption{Qualitative comparison on the \texttt{maze2d-large-v1} task. From left to right: the Diffuser baseline, the truncated diffusion prior, and the final TDP-SafeMAP posterior trajectory. The posterior refinement yields a smoother and safer path while still reaching the goal.}
\label{fig:maze_large}
\end{figure}

The minimum safety distance and total path length of SafeDiffuser are heavily affected
 by the local trap problem. As shown in Fig.~\ref{fig:safediffuser_large}, while the overall direction 
 of its planned path is correct, the local trajectory exhibits distortion or even wall penetration.
  This occurs because its hard CBF constraints, applied in the highly noisy early diffusion stages, 
  enforce trajectory points that appear safe but are geometrically unreasonable. Consequently, these
   points near obstacle boundaries easily get stuck during sampling. In contrast,
 TDP-SafeMAP avoids this issue through a smooth energy penalty and global evolutionary search.

\begin{figure}[!t]
\centering
\includegraphics[width=0.4\columnwidth]{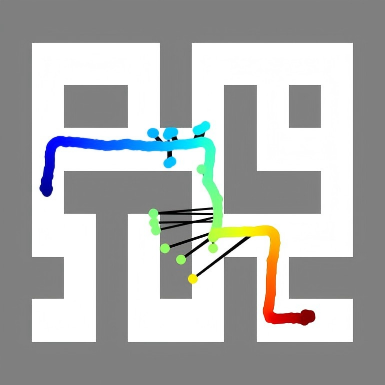}
\caption{Trajectory of SafeDiffuser on the \texttt{maze2d-large-v1} task, illustrating the local trap problem where the path direction is generally correct but exhibits local distortion or wall penetration.}
\label{fig:safediffuser_large}
\end{figure}

To illustrate the mechanism behind the safety improvement, we compare the
trajectory obtained from the mean of the truncated diffusion prior and the
final posterior trajectory refined by CMA-ES. The prior mean already captures
a high-probability mode of the planner distribution but may still pass close
to or through an obstacle; the posterior optimization then relocates the
trajectory in accordance with the requirements of the energy function, reducing the barrier violation while
preserving the task performance.

Fig.~\ref{fig:maze_kappa}(a) depicts the condition number $\kappa(t)$ of the
truncated-prior covariance defined in Section~IV-B: the curve stays low in the
well-conditioned regime and rises steeply as the noise deepens.
Fig.~\ref{fig:maze_kappa}(b) shows its second-order difference, whose clear
peak localizes the elbow of the $\kappa(t)$ curve. The truncation criterion of
Section~IV-B places $t_{\mathrm{stop}}$ right at this elbow: truncating
earlier would leave the Gaussian approximation spanning multiple
undifferentiated modes, whereas truncating later would drive the prior into
the severely ill-conditioned tail. The resulting prior is thus moderately
conditioned, which keeps the subsequent posterior optimization numerically
stable while retaining the mode diversity of the planner distribution.

\begin{figure}[t]
\centering
\begin{subfigure}[b]{\columnwidth}
\includegraphics[width=0.9\linewidth]{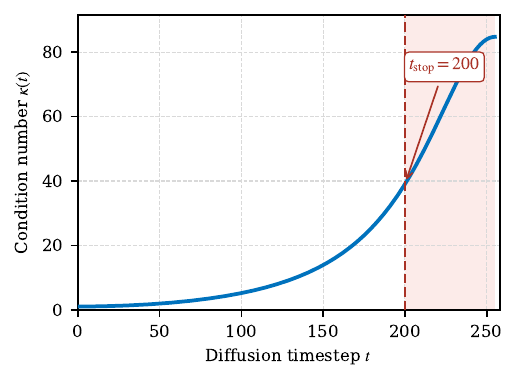}
\caption{Condition number $\kappa(t)$ of the truncated prior covariance; the
vertical line marks the adaptive truncation step $t_{\mathrm{stop}}$. The red region 
indicates the timesteps that need to be executed during actual inference.}
\label{fig:maze_kappa_a}
\end{subfigure}

\vspace{1em}

\begin{subfigure}[b]{\columnwidth}
\includegraphics[width=0.9\linewidth]{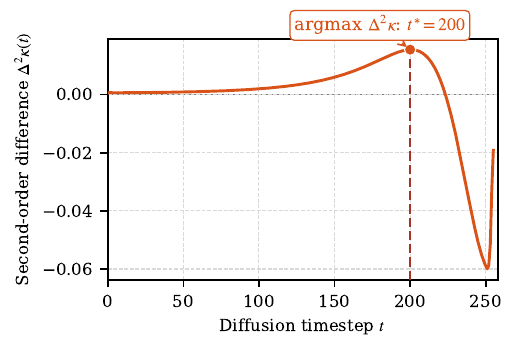}
\caption{Second-order difference of $\kappa(t)$; its peak localizes the
detected onset $t^{\ast}$. The corresponding truncation step is $t_{\mathrm{stop}} = t^{\ast} + \delta$. 
In practice, $\delta$ is typically set to zero.}
\label{fig:maze_kappa_b}
\end{subfigure}
\caption{Condition-number-based adaptive truncation of the diffusion prior. 
It should be noted that larger diffusion timesteps correspond to earlier stages in the inference process.}
\label{fig:maze_kappa}
\end{figure}

\subsection{Locomotion on Hopper}
\label{sec:experiments_hopper}

We evaluate the proposed planner on the locomotion task
\texttt{hopper-medium-expert-v2}, where a single-legged hopper must move
forward while maintaining balance. The planning is performed in a
receding-horizon fashion: at each step, TDP-SafeMAP replans a $32$-step
trajectory, executes its first action, and repeats. The safety constraint is
the ceiling considered in Section~\ref{sec:experiments_setup}. We report the
normalized score, the constraint violation rate, the forward distance and mean top
height. The normalized score and forward distance measure task performance, while the violation rate
and mean top height measure safety. 

\begin{table*}[t]
\caption{Quantitative comparison on the D4RL dataset
\texttt{hopper-medium-expert-v2}. The best result in each column is
highlighted in bold. Higher is better for the normalized score and the forward distance and lower is
better for the violation rate and mean top height.}
\label{tab:hopper}
\centering
\small
\begin{tabular}{lcccc}
\hline
\multirow{2}{*}{Method} & Normalized & Forward & Violation & Mean top \\
       & score & distance (m) & rate (\%) & height (m) \\
\hline
Diffuser & $\bm{0.5128 \pm 0.0004}$ & $\bm{9.3633 \pm 0.0117}$ & $6.7710 \pm 0.1471$ & $1.5793 \pm 0.0019$\\
SafeDiffuser-RoS & $0.4486 \pm 0.0023$ & $7.6909 \pm 0.0631$ & $12.3975 \pm 2.7966$ & $1.6090 \pm 0.0022$\\
TDP-SafeMAP (prior) & $0.4794 \pm 0.0600$ & $8.4920 \pm 1.5643$ & $13.2300 \pm 2.2062$ & $1.7337 \pm 0.1622$\\
TDP-SafeMAP & $0.4841 \pm 0.0629$ & $8.6204 \pm 1.6393$ & $\bm{0.0000 \pm 0.0000}$ & $\bm{1.5663 \pm 0.0035}$\\
\hline
\end{tabular}
\end{table*}

Fig.~\ref{fig:hopper_ceiling} shows the effective body height over the
rollout. The quantitative comparison is summarized in
Table~\ref{tab:hopper}.

It can be observed from Table~\ref{tab:hopper} and Fig.~\ref{fig:hopper_ceiling} 
that only TDP-SafeMAP keeps the body height below the ceiling at every step, achieving zero constraint
 violations in all evaluated rollouts, whereas the baselines exceed the limit. 
Comparing the truncated diffusion prior with the final optimized sequence reveals that even
 if the initial prior planner is unaware of this safety constraint, the subsequent posterior
  refinement can still correct the control actions that would otherwise violate the safety limit.
   Furthermore, the results show that TDP-SafeMAP yields a slightly lower normalized score and a 
   shorter forward distance. This is because constraining the jumping height partially sacrifices 
   the ability to leap forward aggressively to maximize the score, resulting in a slower forward 
   speed. However, this safety improvement only introduces 
a minor trade-off and does not severely compromise the overall task execution.

\begin{figure}[!t]
\centering
\begin{subfigure}[b]{\columnwidth}
\centering
\includegraphics[width=0.85\linewidth]{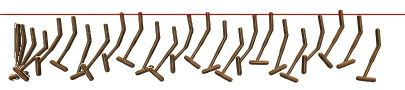}
\caption{Diffuser}
\label{fig:hopper_diffuser}
\end{subfigure}

\vspace{0.5em}

\begin{subfigure}[b]{\columnwidth}
\centering
\includegraphics[width=0.85\linewidth]{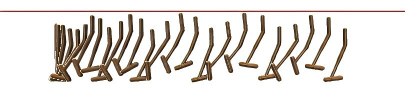}
\caption{TDP-SafeMAP}
\label{fig:hopper_post}
\end{subfigure}
\caption{Effective body height over the rollout on \texttt{hopper-medium-expert-v2}. Diffuser exceeds the $1.8$\,m ceiling at some steps, whereas TDP-SafeMAP keeps the body strictly below it for the entire episode.}
\label{fig:hopper_ceiling}
\end{figure}

We also report the computational cost on the hopper task, measured 
on a 12 vCPU Intel(R) Xeon(R) Silver 4214R CPU at 2.40\,GHz. Using 
the Diffuser baseline as the normalized time unit (i.e., $1.0$), Diffuser,
 SafeDiffuser, and TDP-SafeMAP require $1$, $1.8494$, and $\bm{0.5286}$
  normalized run time, respectively, averaged over five seeds, where the
   CMA-ES population evaluations are distributed over the available cores.
    Even when forced to run serially, TDP-SafeMAP requires only $1.4992$ normalized runtime units, 
    which still outperforms SafeDiffuser. This nearly two-fold 
    speedup over the unconstrained Diffuser primarily stems from the 
    adaptive truncation at $t_{\mathrm{stop}}$ and the parallel evaluation 
    of the CMA-ES population. The speedup will grow further with more cores
     or batched GPU rollouts, 
since the current core count remains well below the population size.

\subsection{Comparison of Truncation Strategies}
\label{sec:experiments_truncation}

The proposed method selects the truncation timestep from the condition-number
curve of the truncated prior covariance (Section~IV-B). To validate this
criterion, we compare it with the speciation-time method~\cite{biroli2024dynamical}, which determines the truncation timestep
from the cumulative-noise schedule alone. Both methods are applied to
TDP-SafeMAP and evaluated on the three maze2d tasks and the hopper task under
identical settings, with the resulting performance summarized in
Table~\ref{tab:truncation_maze} and Table~\ref{tab:truncation_hopper}.

\begin{table}[t]
\caption{Comparison of the truncation strategy on the maze2d tasks.}
\label{tab:truncation_maze}
\centering
\small
\setlength{\tabcolsep}{5pt}
\begin{tabular}{lccc}
\hline
\multirow{2}{*}{Method} & \multirow{2}{*}{$t_{\mathrm{stop}}$} & Minimum safety & Total path\\
     &  & distance (m) & length (m)\\
\hline
\multicolumn{4}{c}{\emph{maze2d-umaze-v1}}\\
Condition (ours) & 48 & $0.4479 \pm 0.0746$ & $3.5978 \pm 0.5290$\\
Speciation       & 49 & $0.4417 \pm 0.0461$ & $3.9043 \pm 0.4190$\\
\multicolumn{4}{c}{\emph{maze2d-medium-v1}}\\
Condition (ours) & 106 & $0.2817 \pm 0.0450$ & $6.8423 \pm 1.0704$\\
Speciation       & 107 & $0.2887 \pm 0.1015$ & $6.9880 \pm 1.2461$\\
\multicolumn{4}{c}{\emph{maze2d-large-v1}}\\
Condition (ours) & 200 & $0.2229 \pm 0.0878$ & $9.8000 \pm 2.3403$\\
Speciation       & 202 & $0.2062 \pm 0.2182$ & $9.8702 \pm 3.0370$\\
\hline
\end{tabular}
\end{table}

\begin{table}[t]
\caption{Comparison of the truncation strategy on the hopper task.}
\label{tab:truncation_hopper}
\centering
\small
\setlength{\tabcolsep}{5pt}
\begin{tabular}{lccc}
\hline
\multirow{2}{*}{Method} & \multirow{2}{*}{$t_{\mathrm{stop}}$} & Violation & Normalized\\
       &  & rate (\%) & score\\
\hline
Condition (ours) & 12 & $0.0000 \pm 0.0000$ & $0.4841 \pm 0.0629$\\
Speciation       & 13 & $0.0000 \pm 0.0000$ & $0.4269 \pm 0.0003$\\
\hline
\end{tabular}
\end{table}

The choice of the truncation timing is critical. Truncating too early leaves
more denoising steps but makes the Gaussian prior deviate further from the
exact planner distribution, degrading the task performance; truncating too late
drives the covariance into the ill-conditioned regime, wasting inference time
and potentially making the prior overly conservative. 

As shown in Table~\ref{tab:truncation_maze} and Table~\ref{tab:truncation_hopper}, 
the timesteps selected by both methods are quite close, yet the speciation-time 
method consistently chooses an earlier truncation step than the condition-number 
criterion. On the maze2d tasks, both methods yield comparable safety performance, 
but the condition-number criterion achieves a shorter total path length. On the 
hopper locomotion task, both methods ensure system safety with zero constraint 
violations, while the condition-number criterion achieves a higher normalized score. 
Compared with the speciation-time method, the condition-number criterion provides 
stronger assurance that the diffusion prior is truncated after mode separation, 
yielding more stable and reliable planning results.

In principle, the condition-number criterion and the speciation-time theory
characterize the emergence of modes in reverse diffusion from different
perspectives. The latter derives a universal time scale for mode separation from
the principal eigenvalue of the data covariance, with a criterion that depends
only on the data spectrum. The former is based on the ratio between the largest
and smallest eigenvalues of the truncated prior covariance $\boldsymbol{\Sigma}_B(t)$,
exploiting information from both ends of the spectrum to characterize whether
the Gaussian approximation has collapsed onto a low-dimensional manifold and
whether it remains in a well-conditioned regime suitable for posterior solving.
Since the speciation-time theory rests on the assumptions of an exact empirical
score, the simultaneous limit of a large number of data points $n$ and a large
data dimension $d$ with $\xi=(\log n)/d$ held fixed, and a single dominant
principal component, we treat it as a spectrum-based truncation heuristic for
comparison. Our method is theoretically better suited to practical settings in 
which the dataset size and dimension do not satisfy the assumptions of the speciation-time theory.

\section{Conclusion}
\label{sec:conclusion}

We proposed TDP-SafeMAP, a Bayesian maximum a posteriori framework
for motion planning. It takes the adaptively truncated output distribution of
diffusion inference as the prior, maps task requirements such as the CBF safety term to a Gibbs 
generalized task likelihood through an exponential energy mapping, and performs
maximum a posteriori estimation with CMA-ES. In this way, TDP-SafeMAP naturally
fuses the expert information learned by the diffusion model with the physical
requirements imposed by the task. It does not require modifying the diffusion
model itself and is therefore plug-and-play. Experiments show that TDP-SafeMAP outperforms
Diffuser and SafeDiffuser in terms of safety and other metrics.

Future work will proceed along the following directions. More efficient prior
acquisition methods, such as the DPM-Solver family or flow matching, can be
adopted to improve inference speed~\cite{11002717}. Variational or Monte Carlo methods can be
used to compute the posterior distribution, thereby providing confidence levels
for the resulting control sequence and additional output information for
improving system safety~\cite{schmal2026reliable}. Moreover, the current framework constructs a Gaussian prior within a selected
diffusion mode. Future work may extend TDP-SafeMAP to mixture priors and
multi-population posterior search, allowing multiple behavior modes to be
preserved and optimized simultaneously. Such an extension could reduce
premature commitment to a suboptimal mode and improve planning robustness in
highly ambiguous environments with multiple topologically distinct safe paths.

\appendices

\section{Proofs for Approximation Error Bound and Safety Preservation}

\subsection{Proof of the Error-Bound Theorem for the Approximate MAP Solution}
\label{app:proof_map_error}

By the assumed CMA-ES suboptimality,
\begin{equation*}
\widehat{F}(\widehat{\mathbf{B}})
\leq
\widehat{F}(\mathbf{B}^{\circ})
+
\zeta_{\mathrm{opt}}.
\end{equation*}
Since
\begin{equation*}
\widehat{F}(\mathbf{B})
=
F^{\circ}(\mathbf{B})
+
\Delta_p(\mathbf{B}),
\end{equation*}
we have
\begin{align*}
F^{\circ}(\widehat{\mathbf{B}})
-
F^{\circ}(\mathbf{B}^{\circ})
\leq\;&
\Delta_p(\mathbf{B}^{\circ})
-
\Delta_p(\widehat{\mathbf{B}})
+
\zeta_{\mathrm{opt}}\\
\leq\;&
\omega_p+\zeta_{\mathrm{opt}}.
\end{align*}
On the other hand, the quadratic-growth condition gives
\begin{equation*}
F^{\circ}(\widehat{\mathbf{B}})
-
F^{\circ}(\mathbf{B}^{\circ})
\geq
\frac{\vartheta_{\mathrm{qg}}}{2}
\left\|
\widehat{\mathbf{B}}
-
\mathbf{B}^{\circ}
\right\|^2.
\end{equation*}
Combining the preceding two inequalities yields
\begin{equation*}
\left\|
\widehat{\mathbf{B}}
-
\mathbf{B}^{\circ}
\right\|^2
\leq
\frac{
2\left(
\omega_p+\zeta_{\mathrm{opt}}
\right)
}{
\vartheta_{\mathrm{qg}}
},
\end{equation*}
which proves
\begin{equation*}
\left\|
\widehat{\mathbf{B}}
-
\mathbf{B}^{\circ}
\right\|
\leq
\sqrt{
\frac{
2\left(
\omega_p+\zeta_{\mathrm{opt}}
\right)
}{
\vartheta_{\mathrm{qg}}
}
}.
\end{equation*}

\subsection{Gaussian Prior Perturbation Bound}

Consider
\begin{equation*}
p_G^{\circ}(\mathbf{B})
=
\mathcal{N}
\left(
\mathbf{B}\mid
\boldsymbol{\mu}_B^{\circ},
\boldsymbol{\Sigma}_B^{\circ}
\right)
\end{equation*}
and
\begin{equation*}
\widehat{p}(\mathbf{B})
=
\mathcal{N}
\left(
\mathbf{B}\mid
\boldsymbol{\mu}_B,
\boldsymbol{\Sigma}_B
\right).
\end{equation*}
Introduce
\begin{equation*}
\mathbf{z}
=
\mathbf{B}-\boldsymbol{\mu}_B^{\circ},
\qquad
\mathbf{v}_\mu
=
\boldsymbol{\mu}_B-\boldsymbol{\mu}_B^{\circ},
\end{equation*}
and
\begin{equation*}
\mathbf{Q}^{\circ}
=
\left(
\boldsymbol{\Sigma}_B^{\circ}
\right)^{-1},
\qquad
\mathbf{Q}
=
\boldsymbol{\Sigma}_B^{-1}.
\end{equation*}
The parameter-induced negative-log-prior error is
\begin{align*}
\Delta_{\mathrm{par}}(\mathbf{B})
=&\,
-\log\widehat{p}(\mathbf{B})
+
\log p_G^{\circ}(\mathbf{B})\\
=&\,
\frac{1}{2}
(\mathbf{z}-\mathbf{v}_\mu)^\top
\mathbf{Q}
(\mathbf{z}-\mathbf{v}_\mu)
-
\frac{1}{2}
\mathbf{z}^\top
\mathbf{Q}^{\circ}
\mathbf{z}
\nonumber\\
&+
\frac{1}{2}
\log
\frac{
\det(\boldsymbol{\Sigma}_B)
}{
\det(\boldsymbol{\Sigma}_B^{\circ})
}.
\end{align*}
After expanding the quadratic term,
\begin{equation*}
\Delta_{\mathrm{par}}(\mathbf{B})
=
\frac{1}{2}
\mathbf{z}^\top
\left(
\mathbf{Q}-\mathbf{Q}^{\circ}
\right)
\mathbf{z}
-
\mathbf{v}_\mu^\top
\mathbf{Q}\mathbf{z}
+
c_{\mathrm{ind}},
\end{equation*}
where
\begin{equation*}
c_{\mathrm{ind}}
=
\frac{1}{2}
\mathbf{v}_\mu^\top
\mathbf{Q}\mathbf{v}_\mu
+
\frac{1}{2}
\log
\frac{
\det(\boldsymbol{\Sigma}_B)
}{
\det(\boldsymbol{\Sigma}_B^{\circ})
}
\end{equation*}
is independent of \(\mathbf{B}\) and therefore does not affect the oscillation.

Define
\begin{equation*}
\varphi(\mathbf{z})
=
\frac{1}{2}
\mathbf{z}^\top
\left(
\mathbf{Q}-\mathbf{Q}^{\circ}
\right)
\mathbf{z}
-
\mathbf{v}_\mu^\top
\mathbf{Q}\mathbf{z}.
\end{equation*}
For
\(\|\mathbf{z}\|\leq R_{\mathrm{mode}}\),
\begin{align*}
\left|\varphi(\mathbf{z})\right|
\leq\;&
\frac{R_{\mathrm{mode}}^2}{2}
\left\|
\mathbf{Q}-\mathbf{Q}^{\circ}
\right\|
\nonumber\\
&+
R_{\mathrm{mode}}
\left\|\mathbf{Q}\right\|
\left\|\mathbf{v}_\mu\right\|.
\end{align*}
Using
\begin{equation*}
\sup\varphi-\inf\varphi
\leq
2\sup|\varphi|
\end{equation*}
and
\(\|\mathbf{v}_\mu\|=\upsilon_\mu\), we obtain
\begin{equation*}
\omega_{\mathrm{par}}
\leq
R_{\mathrm{mode}}^2
\left\|
\mathbf{Q}-\mathbf{Q}^{\circ}
\right\|
+
2R_{\mathrm{mode}}
\left\|\mathbf{Q}\right\|
\upsilon_\mu.
\end{equation*}

By Weyl's eigenvalue perturbation inequality,
\begin{equation*}
\lambda_{\min}
\left(
\boldsymbol{\Sigma}_B
\right)
\geq
s_{\min}^{\circ}-\upsilon_\Sigma.
\end{equation*}
Therefore,
\(\upsilon_\Sigma<s_{\min}^{\circ}\)
ensures that \(\boldsymbol{\Sigma}_B\) is positive definite and
\begin{equation*}
\left\|\mathbf{Q}\right\|
\leq
\frac{1}
{s_{\min}^{\circ}-\upsilon_\Sigma}.
\end{equation*}
The inverse perturbation identity gives
\begin{equation*}
\mathbf{Q}-\mathbf{Q}^{\circ}
=
\boldsymbol{\Sigma}_B^{-1}
\left(
\boldsymbol{\Sigma}_B^{\circ}
-
\boldsymbol{\Sigma}_B
\right)
\left(
\boldsymbol{\Sigma}_B^{\circ}
\right)^{-1},
\end{equation*}
and hence
\begin{equation*}
\left\|
\mathbf{Q}-\mathbf{Q}^{\circ}
\right\|
\leq
\frac{
\upsilon_\Sigma
}{
s_{\min}^{\circ}
\left(
s_{\min}^{\circ}-\upsilon_\Sigma
\right)
}.
\end{equation*}
It follows that
\begin{equation*}
\omega_{\mathrm{par}}
\leq
R_{\mathrm{mode}}^2
\frac{
\upsilon_\Sigma
}{
s_{\min}^{\circ}
\left(
s_{\min}^{\circ}-\upsilon_\Sigma
\right)
}
+
2R_{\mathrm{mode}}
\frac{
\upsilon_\mu
}{
s_{\min}^{\circ}-\upsilon_\Sigma
}.
\end{equation*}

For a non-Gaussian reference prior,
\begin{align*}
-\log\widehat{p}(\mathbf{B})
+
\log p^{\circ}(\mathbf{B})
=&\,
\left[
-\log\widehat{p}(\mathbf{B})
+
\log p_G^{\circ}(\mathbf{B})
\right]
\nonumber\\
&+
\left[
-\log p_G^{\circ}(\mathbf{B})
+
\log p^{\circ}(\mathbf{B})
\right].
\end{align*}
The subadditivity of function oscillations therefore gives
\begin{equation*}
\omega_p
\leq
\omega_{\mathrm{par}}
+
\omega_G,
\end{equation*}
and consequently
\begin{equation*}
\omega_p
\leq
\omega_G
+
R_{\mathrm{mode}}^2
\frac{
\upsilon_\Sigma
}{
s_{\min}^{\circ}
\left(
s_{\min}^{\circ}-\upsilon_\Sigma
\right)
}
+
2R_{\mathrm{mode}}
\frac{
\upsilon_\mu
}{
s_{\min}^{\circ}-\upsilon_\Sigma
}.
\end{equation*}

\subsection{Propagation of the Control-Point Error}

The B\'ezier mapping gives
\begin{equation*}
\widehat{\mathbf{u}}_k
-
\mathbf{u}_k^{\circ}
=
\mathbf{M}_k
\left(
\widehat{\mathbf{B}}
-
\mathbf{B}^{\circ}
\right).
\end{equation*}
Therefore,
\begin{equation*}
\left\|
\widehat{\mathbf{u}}_k
-
\mathbf{u}_k^{\circ}
\right\|
\leq
\left\|\mathbf{M}_k\right\|
\left\|
\widehat{\mathbf{B}}
-
\mathbf{B}^{\circ}
\right\|.
\end{equation*}

Define
\begin{equation*}
e_k
=
\left\|
\widehat{\mathbf{x}}_k
-
\mathbf{x}_k^{\circ}
\right\|.
\end{equation*}
Since the two trajectories start from the same initial state,
\(e_0=0\). The Lipschitz property of the dynamics gives
\begin{equation*}
e_{k+1}
\leq
L_x e_k
+
L_u
\left\|\mathbf{M}_k\right\|
\left\|
\widehat{\mathbf{B}}
-
\mathbf{B}^{\circ}
\right\|.
\end{equation*}
Recursively applying this inequality yields
\begin{equation*}
e_k
\leq
G_k
\left\|
\widehat{\mathbf{B}}
-
\mathbf{B}^{\circ}
\right\|,
\end{equation*}
where
\begin{equation*}
G_k
=
L_u
\sum_{j=0}^{k-1}
L_x^{\,k-1-j}
\left\|\mathbf{M}_j\right\|.
\end{equation*}

For the CBF residual,
\begin{align*}
&
\left|
r_k(\widehat{\mathbf{B}})
-
r_k(\mathbf{B}^{\circ})
\right|
\nonumber\\
&\leq
\left|
h(\widehat{\mathbf{x}}_{k+1})
-
h(\mathbf{x}_{k+1}^{\circ})
\right|
\nonumber\\
&\quad+
(1-\gamma)
\left|
h(\widehat{\mathbf{x}}_k)
-
h(\mathbf{x}_k^{\circ})
\right|.
\end{align*}
The Lipschitz property of \(h\) therefore gives
\begin{align*}
&
\left|
r_k(\widehat{\mathbf{B}})
-
r_k(\mathbf{B}^{\circ})
\right|
\nonumber\\
&\leq
L_h
\left[
G_{k+1}
+
(1-\gamma)G_k
\right]
\left\|
\widehat{\mathbf{B}}
-
\mathbf{B}^{\circ}
\right\|\\
&\leq
K_{\mathrm{CBF}}
\left\|
\widehat{\mathbf{B}}
-
\mathbf{B}^{\circ}
\right\|.
\end{align*}

\subsection{Proof of the CBF Safety-Preservation Theorem}
\label{app:proof_cbf_safety}

The explicit prior-error condition and the Gaussian perturbation bound imply
\begin{equation*}
\omega_p+\zeta_{\mathrm{opt}}
\leq
\frac{
\vartheta_{\mathrm{qg}}
\varrho_{\mathrm{CBF}}^2
}{
2K_{\mathrm{CBF}}^2
}.
\end{equation*}
Applying the error bound for the approximate MAP solution gives
\begin{equation*}
\left\|
\widehat{\mathbf{B}}
-
\mathbf{B}^{\circ}
\right\|
\leq
\frac{
\varrho_{\mathrm{CBF}}
}{
K_{\mathrm{CBF}}
}.
\end{equation*}
Hence, for every \(k=0,\ldots,H-1\),
\begin{align*}
r_k(\widehat{\mathbf{B}})
&\geq
r_k(\mathbf{B}^{\circ})
-
K_{\mathrm{CBF}}
\left\|
\widehat{\mathbf{B}}
-
\mathbf{B}^{\circ}
\right\|\\
&\geq
\varrho_{\mathrm{CBF}}
-
K_{\mathrm{CBF}}
\frac{
\varrho_{\mathrm{CBF}}
}{
K_{\mathrm{CBF}}
}
=0.
\end{align*}
Thus, the discrete-time CBF condition holds at every step:
\begin{equation*}
h(\widehat{\mathbf{x}}_{k+1})
\geq
(1-\gamma)
h(\widehat{\mathbf{x}}_k).
\end{equation*}
Since \(0<\gamma\leq1\) and
\(h(\mathbf{x}_0)\geq0\), induction gives
\begin{equation*}
h(\widehat{\mathbf{x}}_k)
\geq0,
\qquad
k=0,\ldots,H.
\end{equation*}
Therefore, the computed trajectory remains in the safe set over the complete planning horizon.

\section{Derivations for the Unified Statistical Interpretation of TDP-SafeMAP}
\label{app:prior_cma_derivations}

\subsection{Bayesian-Risk Interpretation of Condition-Number Truncation}
\label{app:spectral_truncation}

Because the regularized data covariance is positive definite, write
\begin{equation*}
\boldsymbol{\Sigma}_{\mathrm{data}}
=
\mathbf{Q}
\operatorname{diag}
\left(
\lambda_1^{\mathrm{data}},
\ldots,
\lambda_D^{\mathrm{data}}
\right)
\mathbf{Q}^{\top},
\end{equation*}
where
\begin{equation*}
0<
\lambda_{\min}^{\mathrm{data}}
=
\lambda_1^{\mathrm{data}}
\leq\cdots\leq
\lambda_D^{\mathrm{data}}
=
\lambda_{\max}^{\mathrm{data}}.
\end{equation*}
Let
\begin{equation*}
s_t
=
\rho_t-\rho_{t_{\max}}
\geq0.
\end{equation*}
Iterating Eq.~\eqref{eq:trunc_cov_recursion} from \(t_{\max}\) to \(t\)
gives
\begin{equation*}
\boldsymbol{\Sigma}_{0|t}^{-1}
=
\boldsymbol{\Sigma}_{\mathrm{data}}^{-1}
+s_t\mathbf{I}.
\end{equation*}

Consider the surrogate Gaussian observation model
\begin{equation}
\mathbf{X}
\sim
\mathcal{N}
\left(
\mathbf{0},
\boldsymbol{\Sigma}_{\mathrm{data}}
\right),
\qquad
\mathbf{Y}_t
=
\sqrt{s_t}\mathbf{X}
+
\boldsymbol{\varepsilon},
\qquad
\boldsymbol{\varepsilon}
\sim
\mathcal{N}(\mathbf{0},\mathbf{I}).
\label{eq:surrogate_gaussian_channel}
\end{equation}
The Gaussian conditioning formula gives
\begin{equation*}
\operatorname{Cov}
\left(
\mathbf{X}\mid\mathbf{Y}_t
\right)
=
\left(
\boldsymbol{\Sigma}_{\mathrm{data}}^{-1}
+s_t\mathbf{I}
\right)^{-1}.
\end{equation*}
Therefore, under the Gaussian second-order closure,
\(\boldsymbol{\Sigma}_{0|t}\) is equal to the posterior covariance of
Eq.~\eqref{eq:surrogate_gaussian_channel}. This is a covariance-level
equivalence and does not assume that the original diffusion distribution is
globally Gaussian.

Because the added precision is isotropic, the eigenvectors remain unchanged:
\begin{align*}
\boldsymbol{\Sigma}_{0|t}
&=
\mathbf{Q}
\operatorname{diag}
\left(
v_1(t),\ldots,v_D(t)
\right)
\mathbf{Q}^{\top},\\
v_i(t)
&=
\frac{1}
{1/\lambda_i^{\mathrm{data}}+s_t}
=
\frac{\lambda_i^{\mathrm{data}}}
{1+s_t\lambda_i^{\mathrm{data}}}.
\end{align*}
For squared-error estimation, the posterior mean is the Bayes estimator.
Hence,
\begin{equation*}
v_i(t)
=
\mathbb{E}
\left[
\left(
X_i-\mathbb{E}[X_i\mid Y_{i,t}]
\right)^2
\right]
=
\operatorname{mmse}(X_i\mid Y_{i,t}).
\end{equation*}
Thus, each covariance eigenvalue represents the minimum posterior estimation
risk in its corresponding principal direction.

Since \(\lambda/(1+s\lambda)\) is strictly increasing in \(\lambda\),
\begin{align*}
v_{\max}(t)
&=
\frac{1}
{1/\lambda_{\max}^{\mathrm{data}}+s_t},\\
v_{\min}(t)
&=
\frac{1}
{1/\lambda_{\min}^{\mathrm{data}}+s_t}.
\end{align*}
The condition number is therefore
\begin{align*}
\kappa(t)
&=
\frac{v_{\max}(t)}{v_{\min}(t)}\\
&=
\frac{
1/\lambda_{\min}^{\mathrm{data}}+s_t
}{
1/\lambda_{\max}^{\mathrm{data}}+s_t
}\\
&=
1+
\left(
\frac{1}{\lambda_{\min}^{\mathrm{data}}}
-
\frac{1}{\lambda_{\max}^{\mathrm{data}}}
\right)
v_{\max}(t).
\end{align*}
Defining
\begin{equation*}
c_{\Sigma}
=
\frac{1}{\lambda_{\min}^{\mathrm{data}}}
-
\frac{1}{\lambda_{\max}^{\mathrm{data}}},
\end{equation*}
the central second difference satisfies
\begin{equation}
\Delta^2\kappa(t)
=
c_{\Sigma}\Delta^2v_{\max}(t).
\label{eq:second_difference_mmse}
\end{equation}
When
\(\lambda_{\max}^{\mathrm{data}}>
\lambda_{\min}^{\mathrm{data}}\),
\(c_{\Sigma}>0\), and consequently
\begin{equation*}
\arg\max_t\Delta^2\kappa(t)
=
\arg\max_t\Delta^2v_{\max}(t).
\end{equation*}
The second-difference condition-number criterion is therefore exactly a
curvature criterion for the principal-direction Bayes risk. The maximizing
timestep depends on \(\lambda_{\max}^{\mathrm{data}}\) and the diffusion
schedule, whereas \(\lambda_{\min}^{\mathrm{data}}\) only multiplies the
second-difference curve by the positive factor \(c_{\Sigma}\). If the data
covariance is isotropic, then \(\kappa(t)\equiv1\), and the criterion contains
no unique truncation information.

For a monotone diffusion schedule, let
\begin{equation*}
\delta s_t
=
s_t-s_{t+1}
>0
\end{equation*}
be the precision added by the reverse step \(t+1\rightarrow t\). The
corresponding reduction in the principal-direction Bayes risk is
\begin{align*}
G_t
&=
v_{\max}(t+1)-v_{\max}(t)\\
&=
\frac{\delta s_t}{
\left(
1/\lambda_{\max}^{\mathrm{data}}+s_{t+1}
\right)
\left(
1/\lambda_{\max}^{\mathrm{data}}+s_t
\right)
}\\
&=
\delta s_t\,
v_{\max}(t+1)v_{\max}(t).
\end{align*}
Moreover,
\begin{equation}
\Delta^2v_{\max}(t)
=
G_t-G_{t-1}.
\label{eq:risk_gain_change}
\end{equation}
Combining Eqs.~\eqref{eq:second_difference_mmse} and
\eqref{eq:risk_gain_change} gives
\begin{equation*}
\Delta^2\kappa(t)
=
c_{\Sigma}
\left(
G_t-G_{t-1}
\right).
\end{equation*}
Thus,
\begin{equation*}
t^*
=
\arg\max_t\Delta^2\kappa(t)
\end{equation*}
selects the point at which the statistical value of two adjacent reverse
steps changes most strongly. The criterion is therefore a
schedule-dependent Bayes-risk change point. The offset used by the algorithm,
\begin{equation*}
t_{\mathrm{stop}}
=
\min(t^*+\delta,t_{\max}),
\end{equation*}
acts as a deterministic margin around this detected change point and does not
alter the preceding equivalence for \(t^*\).

The result also has an information-theoretic interpretation. In the
principal scalar Gaussian channel,
\begin{equation*}
I_{\max}(s)
=
\frac{1}{2}
\log
\left(
1+s\lambda_{\max}^{\mathrm{data}}
\right),
\end{equation*}
and
\begin{equation*}
\frac{\mathrm{d}I_{\max}(s)}{\mathrm{d}s}
=
\frac{1}{2}
\frac{\lambda_{\max}^{\mathrm{data}}}
{1+s\lambda_{\max}^{\mathrm{data}}}
=
\frac{1}{2}v_{\max}(s).
\end{equation*}
Hence, \(v_{\max}(s)\) is proportional to the marginal information gain per
unit observation precision. The second-difference criterion detects where
this marginal value changes most sharply along the discrete noise schedule;
it does not maximize the total mutual information over all directions.

To relate this risk to mode selection without assuming a finite Gaussian
mixture, introduce an arbitrary latent mode variable \(Z\). The law of total
covariance gives
\begin{align}
\operatorname{Cov}
\left(
\mathbf{U}\mid\mathbf{Y}_t
\right)
&=
\mathbb{E}
\left[
\operatorname{Cov}
\left(
\mathbf{U}\mid Z,\mathbf{Y}_t
\right)
\middle|
\mathbf{Y}_t
\right]
\nonumber\\
&\quad+
\operatorname{Cov}
\left(
\mathbb{E}
\left[
\mathbf{U}\mid Z,\mathbf{Y}_t
\right]
\middle|
\mathbf{Y}_t
\right).
\label{eq:total_modal_covariance}
\end{align}
The first term measures residual within-mode uncertainty, while the second
measures uncertainty caused by the unresolved mode identity. This
decomposition is valid for continuous or discrete \(Z\), and therefore does
not require a finite Gaussian-mixture representation.

If the between-mode term in
Eq.~\eqref{eq:total_modal_covariance} is concentrated mainly in the leading
eigendirection, \(v_{\max}(t)\) is a second-order proxy for unresolved mode
uncertainty. Under this additional structural condition, \(t^*\) marks the
largest change in the rate of mode-uncertainty reduction. Without that
condition, the exact interpretation remains the principal-direction
Bayes-risk change point, but the condition number alone does not certify
unimodality or complete mode separation.

\subsection{Gaussian Information Projection and B\'ezier Propagation}
\label{app:gaussian_prior_projection}

Let \(\pi_{\mathrm{stop}}^U\) denote the conditional distribution of the
clean control sequence associated with the realized reverse-diffusion branch
at \(t_{\mathrm{stop}}\). Assume that it admits a density with finite
differential entropy and finite second moment, and define
\begin{equation*}
\mathbf{m}_U
=
\mathbb{E}_{\pi_{\mathrm{stop}}^U}[\mathbf{U}],
\qquad
\mathbf{S}_U
=
\operatorname{Cov}_{\pi_{\mathrm{stop}}^U}(\mathbf{U}),
\end{equation*}
with \(\mathbf{S}_U\succ0\).

For an arbitrary Gaussian
\(q=\mathcal{N}(\boldsymbol{\mu},\boldsymbol{\Sigma})\), the terms in the
forward KL divergence that depend on the Gaussian parameters are
\begin{align*}
D_{\mathrm{KL}}
\left(
\pi_{\mathrm{stop}}^U\|q
\right)
&=
c_{\pi}
+
\frac{1}{2}\log\det\boldsymbol{\Sigma}
\nonumber\\
&\quad+
\frac{1}{2}
\mathbb{E}_{\pi_{\mathrm{stop}}^U}
\left[
(\mathbf{U}-\boldsymbol{\mu})^\top
\boldsymbol{\Sigma}^{-1}
(\mathbf{U}-\boldsymbol{\mu})
\right],
\end{align*}
where \(c_{\pi}\) is independent of
\((\boldsymbol{\mu},\boldsymbol{\Sigma})\). The quadratic expectation can be
written as
\begin{align*}
&
\mathbb{E}_{\pi_{\mathrm{stop}}^U}
\left[
(\mathbf{U}-\boldsymbol{\mu})^\top
\boldsymbol{\Sigma}^{-1}
(\mathbf{U}-\boldsymbol{\mu})
\right]
\nonumber\\
&=
\operatorname{tr}
\left(
\boldsymbol{\Sigma}^{-1}\mathbf{S}_U
\right)
+
(\mathbf{m}_U-\boldsymbol{\mu})^\top
\boldsymbol{\Sigma}^{-1}
(\mathbf{m}_U-\boldsymbol{\mu}).
\end{align*}
For fixed \(\boldsymbol{\Sigma}\), this expression is minimized by
\begin{equation*}
\boldsymbol{\mu}
=
\mathbf{m}_U.
\end{equation*}
After substitution, stationarity with respect to
\(\boldsymbol{\Sigma}\) gives
\begin{equation*}
\boldsymbol{\Sigma}^{-1}
-
\boldsymbol{\Sigma}^{-1}
\mathbf{S}_U
\boldsymbol{\Sigma}^{-1}
=
\mathbf{0},
\end{equation*}
whose unique positive-definite solution is
\begin{equation*}
\boldsymbol{\Sigma}
=
\mathbf{S}_U.
\end{equation*}
Consequently,
\begin{equation}
\arg\min_{q\in\mathcal{G}_D}
D_{\mathrm{KL}}
\left(
\pi_{\mathrm{stop}}^U\|q
\right)
=
\mathcal{N}(\mathbf{m}_U,\mathbf{S}_U).
\label{eq:forward_kl_moment_projection}
\end{equation}
Thus, after choosing the Gaussian approximation family, moment matching is
the exact forward-KL information projection rather than an additional
heuristic.

With an exact conditional score, Tweedie's identity returns
\(\mathbf{m}_U\). In the implemented method, the learned noise predictor in
Eq.~\eqref{eq:tweedie} supplies its estimate
\(\boldsymbol{\mu}_U\), while
Eq.~\eqref{eq:trunc_cov_recursion} supplies the Gaussian second-order
estimate \(\boldsymbol{\Sigma}_U\). The implemented approximation is
therefore
\begin{equation*}
\widehat p_U
=
\mathcal{N}
\left(
\boldsymbol{\mu}_U,
\boldsymbol{\Sigma}_U
\right).
\end{equation*}

Now consider the linear B\'ezier transformation
\begin{equation*}
\mathbf{B}
=
\mathbf{A}\mathbf{U},
\end{equation*}
where \(\mathbf{A}\) has full row rank. The Gaussian pushforward is
\begin{equation*}
\widehat p_B
=
\mathcal{N}
\left(
\mathbf{A}\boldsymbol{\mu}_U,
\mathbf{A}\boldsymbol{\Sigma}_U\mathbf{A}^{\top}
\right).
\end{equation*}
Meanwhile, the exact pushforward of
\(\pi_{\mathrm{stop}}^U\) satisfies
\begin{align*}
\mathbb{E}[\mathbf{B}]
&=
\mathbf{A}\mathbf{m}_U,\\
\operatorname{Cov}(\mathbf{B})
&=
\mathbf{A}\mathbf{S}_U\mathbf{A}^{\top}.
\end{align*}
Hence, forward-KL Gaussian moment projection commutes with the linear
B\'ezier transformation: projecting before propagation and projecting after
propagation yield the same Gaussian moments.

The fixed linear map also satisfies
\begin{align*}
\lambda_{\max}
\left(
\boldsymbol{\Sigma}_B
\right)
&\leq
\sigma_{\max}(\mathbf{A})^2
\lambda_{\max}
\left(
\boldsymbol{\Sigma}_U
\right),\\
\lambda_{\min}
\left(
\boldsymbol{\Sigma}_B
\right)
&\geq
\sigma_{\min}(\mathbf{A})^2
\lambda_{\min}
\left(
\boldsymbol{\Sigma}_U
\right).
\end{align*}
It follows that
\begin{equation}
\kappa_2
\left(
\boldsymbol{\Sigma}_B
\right)
\leq
\kappa_2(\mathbf{A})^2
\kappa_2
\left(
\boldsymbol{\Sigma}_U
\right).
\label{eq:bezier_condition_bound}
\end{equation}
Therefore, the diffusion truncation controls the variable part of the prior
conditioning, while the fixed B\'ezier parameterization contributes only the
multiplicative geometric factor \(\kappa_2(\mathbf{A})^2\).

Under the initialization in Algorithm~\ref{alg:cmaes_posterior},
\begin{equation*}
\mathbf{V}^{(0)}
=
(\sigma^{(0)})^2\mathbf{C}^{(0)}
=
(\sigma^{(0)})^2\boldsymbol{\Sigma}_B.
\end{equation*}
Consequently,
\begin{equation*}
\kappa_2
\left(
\mathbf{V}^{(0)}
\right)
=
\kappa_2
\left(
\boldsymbol{\Sigma}_B
\right).
\end{equation*}
The global step size changes the initial search radius but not its
anisotropy. The prior covariance therefore provides the initial
problem-dependent geometry of CMA-ES.

\subsection{Specialization of Complete CMA-ES to the TDP-SafeMAP Posterior}
\label{app:complete_cma_specialization}

For the Gaussian prior in Eq.~\eqref{eq:projected_prior_theory}, the Gibbs
posterior is
\begin{align*}
p_\tau(\mathbf{B}\mid\mathcal{S})
&=
\frac{1}{Z_\tau}
\exp[-\tau J_{\mathrm{task}}(\mathbf{B})]
\mathcal{N}
\left(
\mathbf{B}\mid
\boldsymbol{\mu}_B,
\boldsymbol{\Sigma}_B
\right).
\end{align*}
Taking its negative logarithm gives
\begin{align}
-\log p_\tau(\mathbf{B}\mid\mathcal{S})
&=
\tau J_{\mathrm{task}}(\mathbf{B})
+
\frac{1}{2}
(\mathbf{B}-\boldsymbol{\mu}_B)^\top
\boldsymbol{\Sigma}_B^{-1}
(\mathbf{B}-\boldsymbol{\mu}_B)
\nonumber\\
&\quad+
\log Z_\tau
+
\frac{D'}{2}\log(2\pi)
+
\frac{1}{2}\log\det\boldsymbol{\Sigma}_B\\
&=
f(\mathbf{B})+c_\tau,
\label{eq:posterior_fitness_exact}
\end{align}
where
\begin{equation*}
c_\tau
=
\log Z_\tau
+
\frac{D'}{2}\log(2\pi)
+
\frac{1}{2}\log\det\boldsymbol{\Sigma}_B
\end{equation*}
does not depend on \(\mathbf{B}\). Therefore,
\begin{equation}
\arg\min_{\mathbf{B}}f(\mathbf{B})
=
\arg\max_{\mathbf{B}}
p_\tau(\mathbf{B}\mid\mathcal{S}).
\label{eq:cma_map_target_equivalence}
\end{equation}

More strongly, the complete candidate ordering is preserved. For any two
candidates,
\begin{align*}
f(\mathbf{B}_i)<f(\mathbf{B}_j)
&\Longleftrightarrow
-\log p_\tau(\mathbf{B}_i\mid\mathcal{S})
<
-\log p_\tau(\mathbf{B}_j\mid\mathcal{S})\\
&\Longleftrightarrow
p_\tau(\mathbf{B}_i\mid\mathcal{S})
>
p_\tau(\mathbf{B}_j\mid\mathcal{S}).
\end{align*}
This proves Eq.~\eqref{eq:fitness_posterior_order}.

To express the same result at the population level, define the fitness
quantile under the current CMA-ES search distribution by
\begin{equation*}
R_g(\mathbf{B})
=
\Pr_{\mathbf{B}'\sim q_g^{\mathrm{CMA}}}
\left[
f(\mathbf{B}')\leq f(\mathbf{B})
\right].
\end{equation*}
Using Eq.~\eqref{eq:posterior_fitness_exact},
\begin{equation}
R_g(\mathbf{B})
=
\Pr_{\mathbf{B}'\sim q_g^{\mathrm{CMA}}}
\left[
p_\tau(\mathbf{B}'\mid\mathcal{S})
\geq
p_\tau(\mathbf{B}\mid\mathcal{S})
\right].
\label{eq:posterior_density_quantile}
\end{equation}
Thus, every rank utility of the form
\begin{equation*}
W_g(\mathbf{B})
=
w\!\left(R_g(\mathbf{B})\right)
\end{equation*}
is equivalently a utility of the posterior-density quantile. Neither
\(Z_\tau\) nor the value of the posterior density itself is required.

This equivalence applies to the complete CMA-ES recursion used in the paper.
The weighted mean update and rank-\(\mu\) covariance update depend directly
on the ordered candidates. The covariance evolution path
\(\mathbf{p}_c\) accumulates the corresponding selected mean displacements,
so the rank-one term also depends on the same posterior ordering. Likewise,
the step-size path \(\mathbf{p}_\sigma\) accumulates whitened selected
displacements. Therefore, replacing fitness ranking by decreasing
\(p_\tau(\mathbf{B}\mid\mathcal{S})\) would leave every generated CMA-ES
update unchanged, apart from the treatment of exact ties.

The information-geometric interpretations of these update components are
established results and are not rederived here. For a multivariate Gaussian
search family, the weighted mean and pure rank-\(\mu\) covariance updates
coincide with sample-based Gaussian IGO updates
\cite{akimoto2012theoretical,ollivier2017information}.

The rank-one term in Eq.~\eqref{eq:cmaes_cov} is covered by the MAP-IGO
interpretation of Hamano et al.~\cite{hamano2024natural}. In that
interpretation, a generation-dependent normal-inverse-Wishart prior is
introduced over the search-distribution parameters
\begin{equation*}
\boldsymbol{\eta}_g
=
\left(
\mathbf{m}^{(g)},
\mathbf{C}^{(g)}
\right),
\end{equation*}
and its mean component is placed along the covariance evolution path. The
resulting MAP-IGO covariance recursion reproduces the combination of the
rank-one and rank-\(\mu\) terms in Eq.~\eqref{eq:cmaes_cov}. For a finite
prior-distance parameter \(r\), the same construction additionally produces
the mean increment
\begin{equation}
\Delta\mathbf{m}_{\mathrm{mom}}^{(g)}
=
\frac{c_m c_1}{r c_\mu}
\sigma^{(g)}
\mathbf{p}_c^{(g+1)}.
\label{eq:map_igo_extra_momentum}
\end{equation}
The algorithm in this paper does not use this additional momentum term.
Instead, it corresponds to the non-informative mean-prior limit
\(r\rightarrow\infty\), for which
\begin{equation*}
\Delta\mathbf{m}_{\mathrm{mom}}^{(g)}
\longrightarrow
\mathbf{0},
\end{equation*}
while the rank-one covariance contribution remains in the standard CMA-ES
update. Thus, the complete covariance adaptation used by TDP-SafeMAP,
including its rank-one term, is consistent with the established MAP-IGO
interpretation.

Cumulative step-size adaptation has a different but complementary
statistical role. The path \(\mathbf{p}_\sigma\) accumulates mean
displacements after whitening by the current covariance. Its norm is
compared with
\begin{equation*}
\chi_{D'}
=
\mathbb{E}
\left[
\left\|
\mathcal{N}(\mathbf{0},\mathbf{I}_{D'})
\right\|
\right].
\end{equation*}
A path substantially longer than \(\chi_{D'}\) indicates persistent
alignment of consecutive standardized steps and increases the global search
scale, whereas a shorter path decreases it
\cite{hansen2001completely,HANSEN20115755}. CSA is therefore a standardized
path-length feedback mechanism; it is not claimed to follow from the same
MAP-IGO construction used to interpret the rank-one covariance term.

Finally, the two priors appearing in this discussion must be distinguished.
The diffusion prior
\begin{equation*}
p_{\mathrm{pri}}(\mathbf{B})
=
\mathcal{N}
\left(
\boldsymbol{\mu}_B,
\boldsymbol{\Sigma}_B
\right)
\end{equation*}
is a fixed probability model for the planning variable \(\mathbf{B}\) and
enters the Gibbs posterior. By contrast, the normal-inverse-Wishart prior in
the MAP-IGO interpretation is a generation-dependent auxiliary distribution
over the optimizer parameters
\((\mathbf{m},\mathbf{C})\). It explains the statistical role of the
rank-one update but is not an additional planning prior and is not identified
with \(p_{\mathrm{pri}}\).

The initial member of the search sequence is
\begin{equation*}
q_0^{\mathrm{CMA}}(\mathbf{B})
=
\mathcal{N}
\left(
\mathbf{B}\mid
\boldsymbol{\mu}_B,
(\sigma^{(0)})^2\boldsymbol{\Sigma}_B
\right).
\end{equation*}
It has the same mean and covariance eigenspaces as the diffusion prior, but
\begin{equation*}
q_0^{\mathrm{CMA}}
=
p_{\mathrm{pri}}
\end{equation*}
holds only when \(\sigma^{(0)}=1\). For later generations, the search
distribution is adapted to locate
\begin{equation*}
\mathbf{B}^{*}
=
\arg\max_{\mathbf{B}}
p_\tau(\mathbf{B}\mid\mathcal{S}),
\end{equation*}
but generally
\begin{equation*}
q_g^{\mathrm{CMA}}(\mathbf{B})
\neq
p_\tau(\mathbf{B}\mid\mathcal{S}).
\end{equation*}
Hence, the output mean of CMA-ES is interpreted as a numerical approximation
of the posterior mode, not as a Gaussian approximation of the complete Gibbs
posterior.

\section{Proofs for Prior-Whitened MAP Equivalence}

\subsection{Objective and Mode Equivalence}
\label{app:proof_whitened_map}

From
\(\boldsymbol{\Sigma}_B
=\mathbf{L}_B\mathbf{L}_B^\top\)
and
\(\mathbf{B}
=\boldsymbol{\mu}_B+\mathbf{L}_B\mathbf{b}\),
the Mahalanobis term in Eq.~\eqref{eq:fitness_posterior} becomes
\begin{align*}
&
\left(
\mathbf{B}-\boldsymbol{\mu}_B
\right)^\top
\boldsymbol{\Sigma}_B^{-1}
\left(
\mathbf{B}-\boldsymbol{\mu}_B
\right)
\nonumber\\
&=
\mathbf{b}^\top
\mathbf{L}_B^\top
\mathbf{L}_B^{-\top}
\mathbf{L}_B^{-1}
\mathbf{L}_B
\mathbf{b}
=
\|\mathbf{b}\|^2.
\end{align*}
Therefore,
\begin{equation*}
f
\left(
\boldsymbol{\mu}_B+\mathbf{L}_B\mathbf{b}
\right)
=
\tau
J_{\mathrm{task}}
\left(
\boldsymbol{\mu}_B+\mathbf{L}_B\mathbf{b}
\right)
+
\frac{1}{2}\|\mathbf{b}\|^2
=
f_{\mathrm{w}}(\mathbf{b}).
\end{equation*}
Since \(\mathbf{L}_B\) is nonsingular,
\(\mathcal{T}_B\) is a bijection, which proves
Eq.~\eqref{eq:whitened_map_equivalence}.

For completeness, the transformed posterior density is
\begin{equation*}
p_{\tau,\mathrm{w}}
\left(
\mathbf{b}\mid\mathcal{S}
\right)
=
p_\tau
\left(
\mathcal{T}_B(\mathbf{b})\mid\mathcal{S}
\right)
\left|\det\mathbf{L}_B\right|.
\end{equation*}
Taking the negative logarithm gives
\begin{equation*}
-\log p_{\tau,\mathrm{w}}
\left(
\mathbf{b}\mid\mathcal{S}
\right)
=
-\log p_\tau
\left(
\mathcal{T}_B(\mathbf{b})\mid\mathcal{S}
\right)
-
\log\left|\det\mathbf{L}_B\right|.
\end{equation*}
The Jacobian term is constant in \(\mathbf{b}\); hence the posterior modes correspond exactly. This constant-Jacobian property is essential because MAP modes are not invariant under a general nonlinear reparameterization.

\subsection{Transformation of the Prior and Initial Search Law}

If
\(\mathbf{B}\sim
\mathcal{N}(\boldsymbol{\mu}_B,\boldsymbol{\Sigma}_B)\),
then
\begin{align*}
\mathbb{E}[\mathbf{b}]
&=
\mathbf{L}_B^{-1}
\left(
\mathbb{E}[\mathbf{B}]
-\boldsymbol{\mu}_B
\right)
=
\mathbf{0},\\
\operatorname{Cov}(\mathbf{b})
&=
\mathbf{L}_B^{-1}
\boldsymbol{\Sigma}_B
\mathbf{L}_B^{-\top}
=
\mathbf{I}.
\end{align*}
Thus, the diffusion prior becomes a standard Gaussian in the whitened coordinates.

The CMA-ES sampling covariance is the total covariance
\(\mathbf{V}^{(g)}
=(\sigma^{(g)})^2\mathbf{C}^{(g)}\).
The prior-matching initialization gives
\begin{equation*}
\mathbf{V}^{(0)}
=
\boldsymbol{\Sigma}_B,
\qquad
\mathbf{C}^{(0)}
=
\frac{
\boldsymbol{\Sigma}_B
}{
(\sigma^{(0)})^2
}.
\end{equation*}
Consequently, its whitened total covariance and normalized covariance factor are
\begin{equation*}
\mathbf{V}_{\mathrm{w}}^{(0)}
=
\mathbf{L}_B^{-1}
\mathbf{V}^{(0)}
\mathbf{L}_B^{-\top}
=
\mathbf{I},
\qquad
\mathbf{C}_{\mathrm{w}}^{(0)}
=
\frac{\mathbf{I}}{(\sigma^{(0)})^2}.
\end{equation*}
Therefore, Eq.~\eqref{eq:whitened_initial_search} holds for every
\(\sigma^{(0)}>0\). Changing \(\sigma^{(0)}\) while compensating
\(\mathbf{C}^{(0)}\) changes only the internal factorization, not the initial sampling law.

\subsection{Equivalence of KL-Constrained Search Updates}

Let
\(q_{\mathrm{w}}\)
and
\(q_{\mathrm{w}}'\)
be two Gaussian search densities in the whitened coordinates, and define their pushforwards by
\begin{equation*}
q_B
=
(\mathcal{T}_B)_{\#}q_{\mathrm{w}},
\qquad
q_B'
=
(\mathcal{T}_B)_{\#}q_{\mathrm{w}}'.
\end{equation*}
Using
\(\mathbf{B}=\mathcal{T}_B(\mathbf{b})\)
and
\(\mathrm{d}\mathbf{B}
=|\det\mathbf{L}_B|\mathrm{d}\mathbf{b}\),
we obtain
\begin{align}
D_{\mathrm{KL}}(q_B\|q_B')
&=
\int
q_B(\mathbf{B})
\log
\frac{q_B(\mathbf{B})}{q_B'(\mathbf{B})}
\,\mathrm{d}\mathbf{B}
\nonumber\\
&=
\int
q_{\mathrm{w}}(\mathbf{b})
\log
\frac{
q_{\mathrm{w}}(\mathbf{b})
}{
q_{\mathrm{w}}'(\mathbf{b})
}
\,\mathrm{d}\mathbf{b}
\nonumber\\
&=
D_{\mathrm{KL}}
\left(
q_{\mathrm{w}}\|q_{\mathrm{w}}'
\right).
\label{eq:affine_kl_invariance}
\end{align}

Since
\(f_{\mathrm{w}}(\mathbf{b})
=f(\mathcal{T}_B(\mathbf{b}))\),
the rank of every corresponding sample is unchanged. Therefore, for the rank utility \(W_g\) defined in Section~\ref{sec:cma_compatibility},
\begin{align}
&
\mathbb{E}_{\mathbf{B}\sim q_B}
\left[
W_g(\mathbf{B})
\right]
\nonumber\\
&\qquad=
\mathbb{E}_{\mathbf{b}\sim q_{\mathrm{w}}}
\left[
W_g
\left(
\mathcal{T}_B(\mathbf{b})
\right)
\right].
\label{eq:affine_utility_invariance}
\end{align}
Equations~\eqref{eq:affine_kl_invariance} and
\eqref{eq:affine_utility_invariance} show that both the objective and constraint of the local KL-constrained update are invariant under \(\mathcal{T}_B\). Hence, their Gaussian natural-gradient update directions are mapped into each other by the same affine transformation, proving Eq.~\eqref{eq:whitened_search_equivalence}. This exact statement applies to the rank-based natural-gradient core; equivalence of a complete CMA-ES implementation additionally requires its evolution paths, step-size state, stopping criteria, and numerical safeguards to be transformed consistently.

\subsection{Residual Curvature after Whitening}

Because \(\mathcal{T}_B\) is affine, the chain rule gives
\begin{equation*}
\nabla_{\mathbf{b}}^2
f_{\mathrm{w}}(\mathbf{b})
=
\mathbf{L}_B^\top
\nabla_{\mathbf{B}}^2
f(\mathbf{B})
\mathbf{L}_B.
\end{equation*}
Moreover,
\begin{equation*}
\nabla_{\mathbf{B}}^2f(\mathbf{B})
=
\tau
\nabla_{\mathbf{B}}^2J_{\mathrm{task}}(\mathbf{B})
+
\boldsymbol{\Sigma}_B^{-1}.
\end{equation*}
Since
\begin{equation*}
\mathbf{L}_B^\top
\boldsymbol{\Sigma}_B^{-1}
\mathbf{L}_B
=
\mathbf{I},
\end{equation*}
we obtain Eq.~\eqref{eq:whitened_hessian}.

Let
\begin{equation*}
\mathbf{H}_{\mathrm{task,w}}
=
\mathbf{L}_B^\top
\nabla_{\mathbf{B}}^2J_{\mathrm{task}}
\mathbf{L}_B.
\end{equation*}
If
\(\mathbf{H}_{\mathrm{task,w}}\succeq0\),
then
\begin{equation*}
\lambda_i
\left(
\nabla_{\mathbf{b}}^2f_{\mathrm{w}}
\right)
=
1+\tau
\lambda_i
\left(
\mathbf{H}_{\mathrm{task,w}}
\right).
\end{equation*}
Taking the ratio between the largest and smallest eigenvalues proves
Eq.~\eqref{eq:whitened_condition_number}.

\subsection{Positive-Semidefinite Prior Covariance}

If the unregularized covariance has rank \(r<D'\), write
\begin{equation*}
\boldsymbol{\Sigma}_B
=
\mathbf{L}_{B,r}\mathbf{L}_{B,r}^\top,
\qquad
\mathbf{L}_{B,r}
\in
\mathbb{R}^{D'\times r}.
\end{equation*}
The exact reduced-coordinate formulation is
\begin{equation*}
\mathbf{B}
=
\boldsymbol{\mu}_B
+
\mathbf{L}_{B,r}\mathbf{b}_r,
\qquad
\mathbf{b}_r\in\mathbb{R}^{r}.
\end{equation*}
On this affine support, the prior has an \(r\)-dimensional standard-normal coordinate density and the MAP objective becomes
\begin{equation}
\min_{\mathbf{b}_r\in\mathbb{R}^{r}}
\left[
\tau
J_{\mathrm{task}}
\left(
\boldsymbol{\mu}_B
+
\mathbf{L}_{B,r}\mathbf{b}_r
\right)
+
\frac{1}{2}
\|\mathbf{b}_r\|^2
\right].
\end{equation}
This is the exact extension for a singular prior. Alternatively, one may define
\begin{equation}
\boldsymbol{\Sigma}_{B,\epsilon}
=
\boldsymbol{\Sigma}_B+\epsilon\mathbf{I},
\qquad
\epsilon>0,
\end{equation}
and apply the full-dimensional derivation to
\(\boldsymbol{\Sigma}_{B,\epsilon}\).
The latter is numerically convenient but changes the prior and therefore constitutes a regularized approximation rather than an exact reparameterization of the original singular model.

\bibliographystyle{IEEEtran}
\bibliography{bibliography}

@ARTICLE{Gu2026,
  author={Gu, Zhaoyuan and Li, Junheng and Shen, Wenlan and Yu, Wenhao and Xie, Zhaoming and McCrory, Stephen and Cheng, Xianyi and Shamsah, Abdulaziz and Griffin, Robert and Liu, C. Karen and Kheddar, Abderrahmane and Peng, Xue Bin and Zhu, Yuke and Shi, Guanya and Nguyen, Quan and Cheng, Gordon and Gao, Huijun and Zhao, Ye},
  journal={IEEE/ASME Transactions on Mechatronics},
  title={Humanoid Locomotion and Manipulation: Current Progress and Challenges in Control, Planning, and Learning},
  year={2026},
  volume={31},
  number={2},
  pages={2300--2330},
  doi={10.1109/TMECH.2025.3579247}
}

@article{Hu2026,
  title = {A review of learning-based motion planning: toward a data-driven optimal control approach},
  journal = {Transportation Research Part C: Emerging Technologies},
  volume = {190},
  pages = {105767},
  year = {2026},
  issn = {0968-090X},
  doi = {10.1016/j.trc.2026.105767},
  author = {Hu, Jia and Chang, Yang and Wang, Haoran}
}

@article{HangYin2021,
  title={Modeling, learning, perception, and control methods for deformable object manipulation},
  author={Yin, Hang and Varava, Anastasia and Kragic, Danica},
  journal={Science Robotics},
  volume={6},
  number={54},
  pages={eabd8803},
  year={2021},
  publisher={American Association for the Advancement of Science}
}

@ARTICLE{11244855,
  author={Urain, Julen and Mandlekar, Ajay and Du, Yilun and Muhammad “Mahi” Shafiullah, Nur and Xu, Danfei and Fragkiadaki, Katerina and Chalvatzaki, Georgia and Peters, Jan},
  journal={IEEE Transactions on Robotics}, 
  title={A Survey on Deep Generative Models for Robot Learning From Multimodal Demonstrations}, 
  year={2026},
  volume={42},
  number={},
  pages={60-79},
  doi={10.1109/TRO.2025.3631816}
}

@inproceedings{janner2022diffuser,
  address = {Baltimore, MD, USA},
  title = {Planning with Diffusion for Flexible Behavior Synthesis},
  author = {Michael Janner and Yilun Du and Joshua Tenenbaum and Sergey Levine},
  booktitle = {International Conference on Machine Learning},
  pages = {9902--9915},
  year = {2022}
  
}

@article{chi2025diffusion,
  title={Diffusion policy: Visuomotor policy learning via action diffusion},
  author={Chi, Cheng and Xu, Zhenjia and Feng, Siyuan and Cousineau, Eric and Du, Yilun and Burchfiel, Benjamin and Tedrake, Russ and Song, Shuran},
  journal={The International Journal of Robotics Research},
  volume={44},
  number={10-11},
  pages={1684--1704},
  year={2025},
  publisher={Sage Publications Sage UK: London, England}
}

@article{liu2026diffusion,
  title={Diffusion models in robotics: A survey},
  author={Liu, Xiaokang and Ma, Kevin Yuchen and Gao, Chen and Shou, Mike Zheng},
  journal={International Journal of Computer Vision},
  volume={134},
  number={6},
  pages={299},
  year={2026},
  publisher={Springer}
}

@article{JIN2026115994,
  title={A survey of imitation learning for autonomous driving from an embodied intelligence perspective},
  author={Jin, Chenglong and Zhang, Hua and Zou, Mei and Li, Yang and Lin, Jiawei and Sun, Yanzhao},
  journal={Engineering Applications of Artificial Intelligence},
  volume={182},
  pages={115994},
  year={2026},
  publisher={Elsevier}
}

@ARTICLE{11230786,
  author={Wang, Zengmao and Tang, Qifei and Gao, Wei and Shen, Shuhan},
  journal={IEEE Robotics and Automation Letters}, 
  title={SRIF: A Safer Ranking Inference Framework for Diffusion Policy Models Without Retraining}, 
  year={2026},
  volume={11},
  number={1},
  pages={25-32},
  doi={10.1109/LRA.2025.3629965}
  }

@article{KAZIM2024100931,
  title={Recent advances in path integral control for trajectory optimization: An overview in theoretical and algorithmic perspectives},
  author={Kazim, Muhammad and Hong, JunGee and Kim, Min-Gyeom and Kim, Kwang-Ki K},
  journal={Annual Reviews in Control},
  volume={57},
  pages={100931},
  year={2024},
  publisher={Elsevier}
}

@ARTICLE{10500723,
  author={Zhou, MengChu and Cui, Meiji and Xu, Dian and Zhu, Shuwei and Zhao, Ziyan and Abusorrah, Abdullah},
  journal={IEEE/CAA Journal of Automatica Sinica}, 
  title={Evolutionary Optimization Methods for High-Dimensional Expensive Problems: A Survey}, 
  year={2024},
  volume={11},
  number={5},
  pages={1092-1105},
  doi={10.1109/JAS.2024.124320}
}

@INPROCEEDINGS{11247412,
  author={Lee, Giwon and Park, Daehee and Jeong, Jaewoo and Yoon, Kuk-Jin},
  booktitle={IEEE/RSJ International Conference on Intelligent Robots and Systems (IROS)}, 
  title={Non-differentiable Reward Optimization for Diffusion-based Autonomous Motion Planning}, 
  address = {Hangzhou, China},
  year={2025},
  volume={},
  number={},
  pages={13215-13222},
  doi={10.1109/IROS60139.2025.11247412}
}

@ARTICLE{10057172,
  author={Li, Boqian and Wen, Shiping and Yan, Zheng and Wen, Guanghui and Huang, Tingwen},
  journal={IEEE/CAA Journal of Automatica Sinica}, 
  title={A Survey on the Control Lyapunov Function and Control Barrier Function for Nonlinear-Affine Control Systems}, 
  year={2023},
  volume={10},
  number={3},
  pages={584-602},
  doi={10.1109/JAS.2023.123075}
}

@ARTICLE{11267490,
  author={Zhao, Shanshan and Yan, Zheng and Huang, Tingwen and Wen, Shiping},
  journal={IEEE Transactions on Cybernetics}, 
  title={A Comprehensive Review on Control Barrier Functions: Uncertainty Handling, Design Optimization, and Feasibility Analysis}, 
  year={2026},
  volume={56},
  number={4},
  pages={2061-2069},
  doi={10.1109/TCYB.2025.3633800}
}

@ARTICLE{10171460,
  author={Zhang, Junhui and Ding, Lizhi and Lu, Xiaonan and Tang, Wenyuan},
  journal={IEEE Transactions on Industry Applications}, 
  title={A Novel Real-Time Control Approach for Sparse and Safe Frequency Regulation in Inverter Intensive Microgrids}, 
  year={2023},
  volume={59},
  number={5},
  pages={5550-5558},
  doi={10.1109/TIA.2023.3291353}
}

@ARTICLE{7782377,
  author={Ames, Aaron D. and Xu, Xiangru and Grizzle, Jessy W. and Tabuada, Paulo},
  journal={IEEE Transactions on Automatic Control}, 
  title={Control Barrier Function Based Quadratic Programs for Safety Critical Systems}, 
  year={2017},
  volume={62},
  number={8},
  pages={3861-3876},
  doi={10.1109/TAC.2016.2638961}
}

@inproceedings{agrawal2017discrete,
  title={Discrete control barrier functions for safety-critical control of discrete systems with application to bipedal robot navigation.},
  author={Agrawal, Ayush and Sreenath, Koushil},
  booktitle={Robotics: Science and systems},
  pages={1--10},
  year={2017},
  address={Cambridge, MA, USA}
}

@ARTICLE{11007767,
  author={Zhang, Junhui and Hoxha, Bardh and Fainekos, Georgios and Panagou, Dimitra},
  journal={IEEE Control Systems Letters}, 
  title={Conformal Prediction in the Loop: Risk-Aware Control Barrier Functions for Stochastic Systems With Data-Driven State Estimators}, 
  year={2025},
  volume={9},
  number={},
  pages={282-287},
  doi={10.1109/LCSYS.2025.3571828}
}

@article{AN2026112782,
  title={CBF-based safety design for adaptive control of uncertain nonlinear strict-feedback systems},
  author={An, Liwei and Zhao, Can and Yang, Guang-Hong},
  journal={Automatica},
  volume={185},
  pages={112782},
  year={2026},
  publisher={Elsevier}
}

@INPROCEEDINGS{9483029,
  author={Zeng, Jun and Zhang, Bike and Sreenath, Koushil},
  booktitle={2021 American Control Conference (ACC)}, 
  title={Safety-Critical Model Predictive Control with Discrete-Time Control Barrier Function}, 
  year={2021},
  volume={},
  number={},
  pages={3882-3889},
  address={Baltimore, MD, USA},
}

@INPROCEEDINGS{10644741,
  author={Abdi, Hossein and Zhao, Pan and Hovakimyan, Naira and Ghabcheloo, Reza},
  booktitle={2024 American Control Conference (ACC)}, 
  title={Model Predictive Control Barrier Functions: Guaranteed Safety with Reduced Conservatism and Shortened Horizon}, 
  year={2024},
  volume={},
  number={},
  address={Toronto, ON, Canada},
  pages={1652-1657},
  doi={10.23919/ACC60939.2024.10644741}
}

@ARTICLE{10539292,
  author={Ma, Qian and Jin, Peng and Lewis, Frank L.},
  journal={IEEE/CAA Journal of Automatica Sinica}, 
  title={Guaranteed Cost Attitude Tracking Control for Uncertain Quadrotor Unmanned Aerial Vehicle Under Safety Constraints}, 
  year={2024},
  volume={11},
  number={6},
  pages={1447-1457},
  doi={10.1109/JAS.2024.124317}
}

@ARTICLE{11606602,
  author={Wang, Xinyang and Zhang, Hongwei and Xu, Jun and Wang, Shimin and Guay, Martin},
  journal={IEEE Transactions on Neural Networks and Learning Systems}, 
  title={Reinforcement Learning in Pursuit–Evasion Differential Game: Safety, Stability, and Robustness}, 
  year={2026},
  volume={},
  number={},
  pages={1-14},
  doi={10.1109/TNNLS.2026.3710199}
}

@ARTICLE{11578271,
  author={Wang, Xinyang and Zhang, Hongwei and Wang, Shimin and Xiao, Wei and Guay, Martin},
  journal={IEEE Transactions on Automatic Control}, 
  title={Safe Learning Control with Optimality and Stability Guarantees}, 
  year={2026},
  volume={},
  number={},
  pages={1-16},
  doi={10.1109/TAC.2026.3707508}
}

@inproceedings{xiao2025safe,
  title = {SafeDiffuser: Safe Planning with Diffusion Probabilistic Models},
  author = {Wei Xiao and Tsun-Hsuan Wang and Chuang Gan and Ramin Hasani and Mathias Lechner and Daniela Rus},
  booktitle = {International Conference on Learning Representations},
  year = {2025},
  address = {Singapore},
  pages = {100391--100412},
}

@inproceedings{MizutaLeung2024,
  title = {CoBL-Diffusion: Diffusion-Based Conditional Robot Planning in Dynamic Environments Using Control Barrier and Lyapunov Functions},
  author = {Kazuki Mizuta and Karen Leung},
  booktitle = {IEEE/RSJ International Conference on Intelligent Robots and Systems (IROS)},
  year = {2024},
  address = {Abu Dhabi, United Arab Emirates},
  pages = {13801--13808}
}

@inproceedings{ku2026pc,
  title={PC-Diffuser: Path-Consistent Capsule CBF Safety Filtering for Diffusion-Based Trajectory Planner},
  author={Ku, Eugene and Lyu, Yiwei},
  booktitle={IEEE/RSJ International Conference on Intelligent Robots and Systems (IROS)},
  year={2026},
  address={Pittsburgh, PA, USA}
}

@inproceedings{NEURIPS2025_31fb284a,
  title={Constrained diffusers for safe planning and control},
  author={Zhang, Jichen and Zhao, Liqun and Papachristodoulou, Antonis and Umenberger, Jack},
  journal={Advances in Neural Information Processing Systems},
  volume={38},
  booktitle={The Thirty-ninth Annual Conference on Neural Information Processing Systems},
  pages={34965--34998},
  year={2025},
  address = {San Diego, CA, USA}
}

@article{GARG2024100945,
  title={Advances in the theory of control barrier functions: Addressing practical challenges in safe control synthesis for autonomous and robotic systems},
  author={Garg, Kunal and Usevitch, James and Breeden, Joseph and Black, Mitchell and Agrawal, Devansh and Parwana, Hardik and Panagou, Dimitra},
  journal={Annual Reviews in Control},
  volume={57},
  pages={100945},
  year={2024},
  publisher={Elsevier}
}

@ARTICLE{9681233,
  author={Qin, Zengyi and Sun, Dawei and Fan, Chuchu},
  journal={IEEE Robotics and Automation Letters}, 
  title={Sablas: Learning Safe Control for Black-Box Dynamical Systems}, 
  year={2022},
  volume={7},
  number={2},
  pages={1928-1935},
  doi={10.1109/LRA.2022.3142743}
}

@ARTICLE{11097366,
  author={Carvalho, João and Le, An Thai and Kicki, Piotr and Koert, Dorothea and Peters, Jan},
  journal={IEEE Transactions on Robotics}, 
  title={Motion Planning Diffusion: Learning and Adapting Robot Motion Planning With Diffusion Models}, 
  year={2025},
  volume={41},
  number={},
  pages={4881-4901},
  doi={10.1109/TRO.2025.3593109}
}

@ARTICLE{11039694,
  author={Alyaseen, Mohammed and Atanasov, Nikolay and Cortés, Jorge},
  journal={IEEE Control Systems Letters}, 
  title={Safe Control of Second-Order Systems With Linear Positional Constraints}, 
  year={2025},
  volume={9},
  number={},
  pages={1279-1284},
  doi={10.1109/LCSYS.2025.3580776}
}

@article{Ghahramani2015,
  title={Probabilistic machine learning and artificial intelligence},
  author={Ghahramani, Zoubin},
  journal={Nature},
  volume={521},
  number={7553},
  pages={452--459},
  year={2015},
  publisher={Nature Publishing Group UK London}
}

@article{biroli2024dynamical,
  title={Dynamical regimes of diffusion models},
  author={Biroli, Giulio and Bonnaire, Tony and De Bortoli, Valentin and M{\'e}zard, Marc},
  journal={Nature Communications},
  volume={15},
  number={1},
  pages={9957},
  year={2024},
  publisher={Nature Publishing Group UK London}
}

@inproceedings{diffusiondrive,
  title     = {DiffusionDrive: Truncated Diffusion Model for End-to-End Autonomous Driving},
  author    = {Liao, Bencheng and Chen, Shaoyu and Yin, Haoran and Jiang, Bo and Wang, Cheng and Yan, Sixu and Zhang, Xinbang and Li, Xiangyu and Zhang, Ying and Zhang, Qian and others},
  booktitle = {2025 IEEE/CVF Conference on Computer Vision and Pattern Recognition (CVPR)},
  pages     = {12037--12047},
  address = {Nashville, TN, USA},
  year      = {2025},
}

@article{hansen2001completely,
  title={Completely derandomized self-adaptation in evolution strategies},
  author={Hansen, Nikolaus and Ostermeier, Andreas},
  journal={Evolutionary computation},
  volume={9},
  number={2},
  pages={159--195},
  year={2001},
  publisher={MIT Press}
}

@article{HANSEN20115755,
  title={Impacts of invariance in search: When CMA-ES and PSO face ill-conditioned and non-separable problems},
  author={Hansen, Nikolaus and Ros, Raymond and Mauny, Nikolas and Schoenauer, Marc and Auger, Anne},
  journal={Applied Soft Computing},
  volume={11},
  number={8},
  pages={5755--5769},
  year={2011},
  publisher={Elsevier}
}

@article{aguirre2026covariance,
  title={Covariance Matrix Adaptation Evolution Strategy (CMA-ES): A Comprehensive Survey of Variants and Hybridizations},
  author={Aguirre, Nahum and Cuevas, Erik and Luque-Chang, Alberto and Barba-Toscano, Oscar and Vasquez-Franco, Mario},
  journal={Archives of Computational Methods in Engineering},
  pages={1--57},
  year={2026},
  publisher={Springer}
}

@article{akimoto2012theoretical,
  title={Theoretical foundation for CMA-ES from information geometry perspective},
  author={Akimoto, Youhei and Nagata, Yuichi and Ono, Isao and Kobayashi, Shigenobu},
  journal={Algorithmica},
  volume={64},
  number={4},
  pages={698--716},
  year={2012},
  publisher={Springer}
}

@article{ollivier2017information,
  title={Information-geometric optimization algorithms: A unifying picture via invariance principles},
  author={Ollivier, Yann and Arnold, Ludovic and Auger, Anne and Hansen, Nikolaus},
  journal={Journal of Machine Learning Research},
  volume={18},
  number={18},
  pages={1--65},
  year={2017}
}

@inproceedings{hamano2024natural,
  title={Natural gradient interpretation of rank-one update in CMA-ES},
  author={Hamano, Ryoki and Shirakawa, Shinichi and Nomura, Masahiro},
  booktitle={International Conference on Parallel Problem Solving from Nature},
  pages={252--267},
  address = {Hagenberg, Austria},
  year={2024},
}

@article{efron2011tweedie,
  title={Tweedie’s formula and selection bias},
  author={Efron, Bradley},
  journal={Journal of the American Statistical Association},
  volume={106},
  number={496},
  pages={1602--1614},
  year={2011},
  publisher={Taylor \& Francis}
}

@inproceedings{rissanen2025free,
  title={Free hunch: Denoiser covariance estimation for diffusion models without extra costs},
  author={Rissanen, Severi and Heinonen, Markus and Solin, Arno},
  booktitle={International Conference on Learning Representations},
  address = {Singapore},
  pages={90079--90112},
  year={2025}
}

@inproceedings{nichol2021improved,
  title={Improved denoising diffusion probabilistic models},
  author={Nichol, Alexander Quinn and Dhariwal, Prafulla},
  booktitle={International conference on machine learning},
  address = {Vienna, Austria},
  pages={8162--8171},
  year={2021},
}

@article{ho2020denoising,
  title={Denoising diffusion probabilistic models},
  author={Ho, Jonathan and Jain, Ajay and Abbeel, Pieter},
  journal={Advances in neural information processing systems},
  volume={33},
  pages={6840--6851},
  year={2020}
}

@ARTICLE{8588399,
  author={Zhang, Cheng and Bütepage, Judith and Kjellström, Hedvig and Mandt, Stephan},
  journal={IEEE Transactions on Pattern Analysis and Machine Intelligence}, 
  title={Advances in Variational Inference}, 
  year={2019},
  volume={41},
  number={8},
  pages={2008-2026},
  doi={10.1109/TPAMI.2018.2889774}
}

@article{van2021bayesian,
  title={Bayesian statistics and modelling},
  author={Van de Schoot, Rens and Depaoli, Sarah and King, Ruth and Kramer, Bianca and M{\"a}rtens, Kaspar and Tadesse, Mahlet G and Vannucci, Marina and Gelman, Andrew and Veen, Duco and Willemsen, Joukje and others},
  journal={Nature Reviews Methods Primers},
  volume={1},
  number={1},
  pages={1},
  year={2021},
  publisher={Nature Publishing Group UK London}
}

@article{farouki2012bernstein,
  title={The Bernstein polynomial basis: A centennial retrospective},
  author={Farouki, Rida T},
  journal={Computer Aided Geometric Design},
  volume={29},
  number={6},
  pages={379--419},
  year={2012},
  publisher={Elsevier}
}

@inproceedings{ames2019control,
  title={Control barrier functions: Theory and applications},
  author={Ames, Aaron D and Coogan, Samuel and Egerstedt, Magnus and Notomista, Gennaro and Sreenath, Koushil and Tabuada, Paulo},
  booktitle={2019 18th European control conference (ECC)},
  pages={3420--3431},
  address = {Naples, Italy},
  year={2019},
}

@article{Wu_2024, title={Convergence Analysis of an Adaptively Regularized Natural Gradient Method}, volume={72}, DOI={10.1109/tsp.2024.3398496}, journal={IEEE Transactions on Signal Processing}, publisher={Institute of Electrical and Electronics Engineers (IEEE)}, author={Wu, Jiayuan and Hu, Jiang and Zhang, Hongchao and Wen, Zaiwen}, year={2024}, pages={2527–2542} }

@inproceedings{NEURIPS2023_893a5db6,
 author = {Zhou, Siyuan and Du, Yilun and Zhang, Shun and Xu, Mengdi and Shen, Yikang and Xiao, Wei and Yeung, Dit-Yan and Gan, Chuang},
 booktitle = {Advances in Neural Information Processing Systems},
 doi = {10.52202/075280-1909},
 editor = {A. Oh and T. Naumann and A. Globerson and K. Saenko and M. Hardt and S. Levine},
 pages = {44000--44016},
 publisher = {Curran Associates, Inc.},
 title = {Adaptive Online Replanning with Diffusion Models},
 volume = {36},
 year = {2023},
 address = {New Orleans, LA, USA}
}

@ARTICLE{9779571,
  author={Wabersich, Kim P. and Zeilinger, Melanie N.},
  journal={IEEE Transactions on Automatic Control}, 
  title={Predictive Control Barrier Functions: Enhanced Safety Mechanisms for Learning-Based Control}, 
  year={2023},
  volume={68},
  number={5},
  pages={2638-2651},
  doi={10.1109/TAC.2022.3175628}}

@ARTICLE{10872802,
  author={Daş, Ersin and Burdick, Joel W.},
  journal={IEEE Transactions on Automatic Control}, 
  title={Robust Control Barrier Functions Using Uncertainty Estimation With Application to Mobile Robots}, 
  year={2025},
  volume={70},
  number={7},
  pages={4766-4773},
  doi={10.1109/TAC.2025.3538742}}

@ARTICLE{10328440,
  author={Yang, Yujie and Zhang, Yuhang and Zou, Wenjun and Chen, Jianyu and Yin, Yuming and Eben Li, Shengbo},
  journal={IEEE Transactions on Automatic Control}, 
  title={Synthesizing Control Barrier Functions With Feasible Region Iteration for Safe Reinforcement Learning}, 
  year={2024},
  volume={69},
  number={4},
  pages={2713-2720},
  doi={10.1109/TAC.2023.3336263}}

@ARTICLE{11002717,
  author={Ma, Zhiyuan and Zhang, Yuzhu and Jia, Guoli and Zhao, Liangliang and Ma, Yichao and Ma, Mingjie and Liu, Gaofeng and Zhang, Kaiyan and Ding, Ning and Li, Jianjun and Zhou, Bowen},
  journal={IEEE Transactions on Pattern Analysis and Machine Intelligence}, 
  title={Efficient Diffusion Models: A Comprehensive Survey From Principles to Practices}, 
  year={2025},
  volume={47},
  number={9},
  pages={7506-7525},
  doi={10.1109/TPAMI.2025.3569700}}

@article{schmal2026reliable,
  title = {Reliable uncertainty estimates in deep learning with efficient {M}etropolis-{H}astings algorithms},
  author = {Schmal, Matthias and M{\"a}der, Patrick},
  journal = {Nature Communications},
  volume = {17},
  pages = {2531},
  year = {2026},
  doi = {10.1038/s41467-026-70015-z}
}

@article{10.1145/3583136,
author = {Guo, Huihui and Wu, Fan and Qin, Yunchuan and Li, Ruihui and Li, Keqin and Li, Kenli},
title = {Recent Trends in Task and Motion Planning for Robotics: A Survey},
year = {2023},
issue_date = {December 2023},
publisher = {Association for Computing Machinery},
address = {New York, NY, USA},
volume = {55},
number = {13s},
issn = {0360-0300},
doi = {10.1145/3583136},
journal = {ACM Computing Surveys},
articleno = {289},
numpages = {36}
}

\end{document}